\documentclass[12pt]{amsart} 

\usepackage[letterpaper,margin=1in]{geometry}

\usepackage{amssymb,amsfonts,amsmath,latexsym,amsthm}
\usepackage[usenames]{color}
\usepackage[]{graphicx}
\usepackage[space]{grffile}
\usepackage{mathrsfs}   
\usepackage{caption}
\usepackage{subcaption}
\usepackage{verbatim}
\usepackage{url}
\usepackage{bm}
\usepackage{extarrows}
\usepackage{dsfont}
\usepackage{epstopdf}
\usepackage[colorlinks,citecolor=black,linkcolor=black,urlcolor=black]{hyperref}

\usepackage[authoryear,round]{natbib}

\theoremstyle{plain}
\newtheorem{theorem}{Theorem}[section]
\newtheorem{corollary}[theorem]{Corollary}
\newtheorem{lemma}[theorem]{Lemma}
\newtheorem{proposition}[theorem]{Proposition}
\newtheorem{condition}[theorem]{Condition}

\theoremstyle{definition}
\newtheorem{definition}[theorem]{Definition}
\numberwithin{equation}{section}

\newcommand{\PX}{\hat P_{X_{1:n}}}

\newcommand{\Px}{\hat P_{x_{1:n}}}
\newcommand{\btheta}{{\bm\theta}}
\newcommand{\bbtheta}{{\pmb{\bm\theta}}}

\renewcommand{\epsilon}{\varepsilon}

\newcommand{\Pb}{{\bm P}}
\newcommand{\s}{{\bm s}}
\newcommand{\q}{{\bm q}}
\newcommand{\HH}{\mathrm{H}}
\newcommand{\SN}{\mathcal{SN}}

\DeclareMathOperator*{\Exponential}{Exp}
\DeclareMathOperator*{\Exp}{Exp}
\DeclareMathOperator*{\Bernoulli}{Bernoulli}
\DeclareMathOperator*{\Uniform}{Uniform}
\DeclareMathOperator*{\Binomial}{Binomial}
\DeclareMathOperator*{\Beta}{Beta}
\DeclareMathOperator*{\BetaBinomial}{BetaBinomial}
\DeclareMathOperator*{\Ga}{Gamma}

\DeclareMathOperator*{\argmin}{argmin}
\DeclareMathOperator*{\argmax}{argmax}
\DeclareMathOperator*{\Cov}{Cov}
\DeclareMathOperator*{\diag}{diag}

\newcommand{\T}{\mathtt{T}}
\newcommand{\R}{\mathbb{R}}

\newcommand{\E}{\mathbb{E}}
\renewcommand{\Pr}{\mathbb{P}}
\newcommand{\I}{\mathds{1}}

\newcommand{\A}{\mathcal{A}}

\newcommand{\D}{\mathcal{D}}
\newcommand{\Hcal}{\mathcal{H}}
\newcommand{\M}{\mathcal{M}}
\newcommand{\N}{\mathcal{N}}
\newcommand{\X}{\mathcal{X}}

\newcommand{\iid}{\text{ i.i.d.}\sim}

\def\app#1#2{%
  \mathrel{%
    \setbox0=\hbox{$#1\sim$}%
    \setbox2=\hbox{%
      \rlap{\hbox{$#1\propto$}}%
      \lower1.3\ht0\box0%
    }%
    \raise0.25\ht2\box2%
  }%
}
\def\approxprop{\mathpalette\app\relax}

\title{Robust Bayesian inference via coarsening}
\author{Jeffrey W. Miller \and David B. Dunson}
\dedicatory{Duke University, Department of Statistical Science}

\begin{document}
\begin{abstract}
The standard approach to Bayesian inference is based on the assumption that the distribution of the data belongs to the chosen model class. However, even a small violation of this assumption can have a large impact on the outcome of a Bayesian procedure. We introduce a simple, coherent approach to Bayesian inference that improves robustness to perturbations from the model: rather than condition on the data exactly, one conditions on a neighborhood of the empirical distribution.  When using neighborhoods based on relative entropy estimates, the resulting ``coarsened'' posterior can be approximated by simply tempering the likelihood---that is, by raising it to a fractional power---thus, inference is often easily implemented with standard methods, and one can even obtain analytical solutions when using conjugate priors. Some theoretical properties are derived, and we illustrate the approach with real and simulated data, using mixture models, autoregressive models of unknown order, and variable selection in linear regression.

\vspace{1em}
\noindent \keywordsname. 
Bayesian inference, Mixture model, Model misspecification, Relative entropy, Robustness, Tempering.
\end{abstract}
\maketitle

\section{Introduction}
\label{section:introduction}


In many applications, the most natural models are idealizations that are known to provide only an
approximation to the distribution of the observed data, due to small-scale contaminating effects that may be
complicated and not completely understood.  
One might hope that any such lack of model fit, if sufficiently small, would not significantly impact
inferences or decisions made based on the model.
Often this does seem to be the case, but sometimes, unfortunately, the likelihood is strongly affected 
by perturbations to the distribution of the observed data, especially when the sample size is large.

As a result, standard Bayesian procedures are not generally robust to contamination or misspecification. 
In particular, this may lead to underestimation of uncertainty, since the posterior
inexorably concentrates at a given rate---typically at minimal Kullback--Leibler points---regardless
of whether or not it is concentrating on something that resembles the observed data distribution.


This issue was, perhaps, not as severe in the past, since on small datasets, 
statistical models automatically exhibit a certain amount of robustness to misspecification,
because, in essence, there is insufficient power to discern small departures from the model. However, as
datasets grow ever larger, even slight deficiencies in our models can become disruptive. 

The problem is especially apparent in the context of Bayesian model averaging and flexible Bayesian models,
such as nonparametric models. Under misspecification, model averaging tends to favor more complex models as the sample size grows,
even when for all practical purposes the data are very well-described by a simpler model. 
Since it is usually unreasonable to expect real data to come exactly from a simple
parametric model, we thus find ourselves in the awkward situation of eventually rejecting any such model.

At first, it may seem paradoxical that nonparametric models could suffer from misspecification issues,
since after all, the point of a nonparametric model is that it can fit any distribution.
It is true that if one is solely interested in fitting the data distribution---such as in
density estimation or when the distribution is a nuisance parameter---then a nonparametric model will
serve nicely, in principle.
However, most nonparametric models still involve some parametric assumptions (such as Gaussian
mixture components) and when the data do not accord well with these assumptions,
the interpretability of parameters and latent variables (such as component parameters and cluster assignments) breaks down.
Further, the computational burden of nonparametric models can grow rapidly with the sample size.

\begin{figure}
  \centering
  \includegraphics[trim=.5cm 0 1.75cm 0, clip, width=0.49\textwidth]{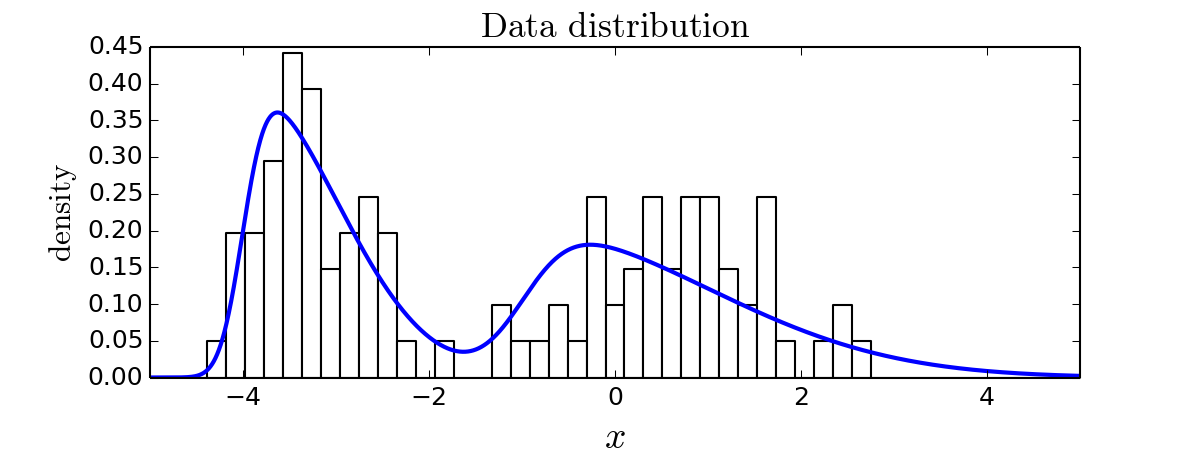} \\
  \includegraphics[trim=.5cm 0 1.75cm 0, clip, width=0.49\textwidth]{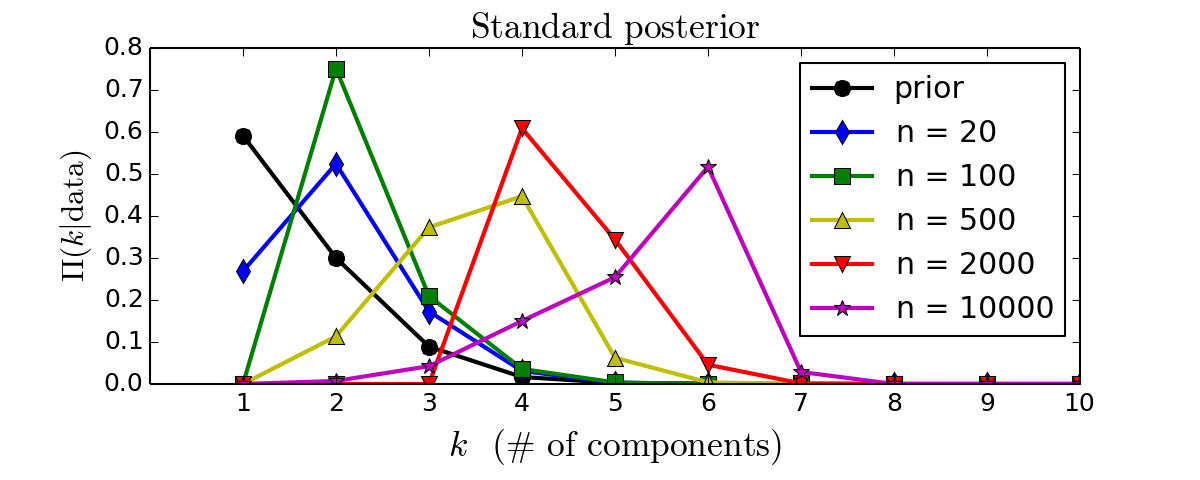}
  \includegraphics[trim=.5cm 0 1.75cm 0, clip, width=0.49\textwidth]{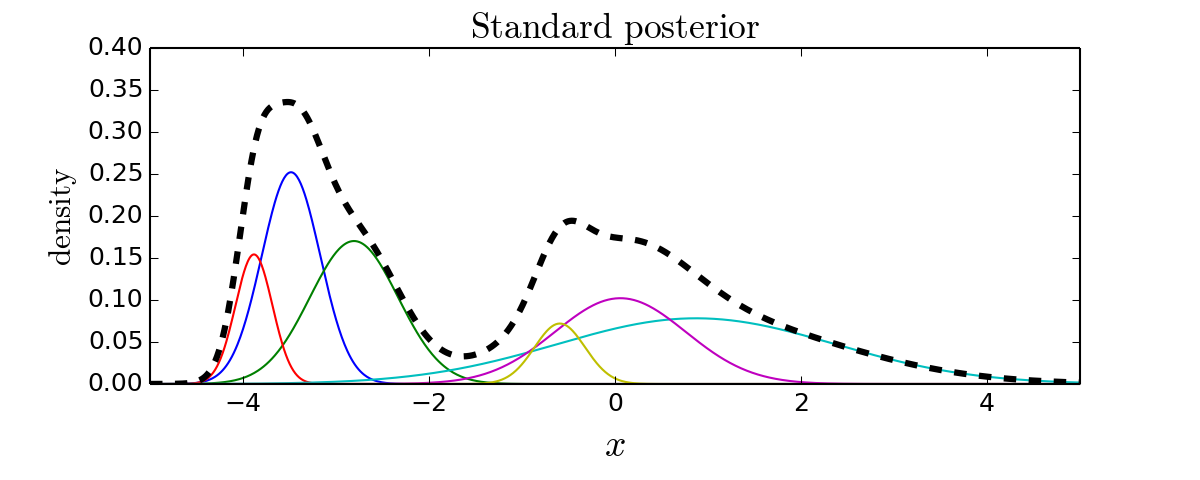}
  \includegraphics[trim=.5cm 0 1.75cm 0, clip, width=0.49\textwidth]{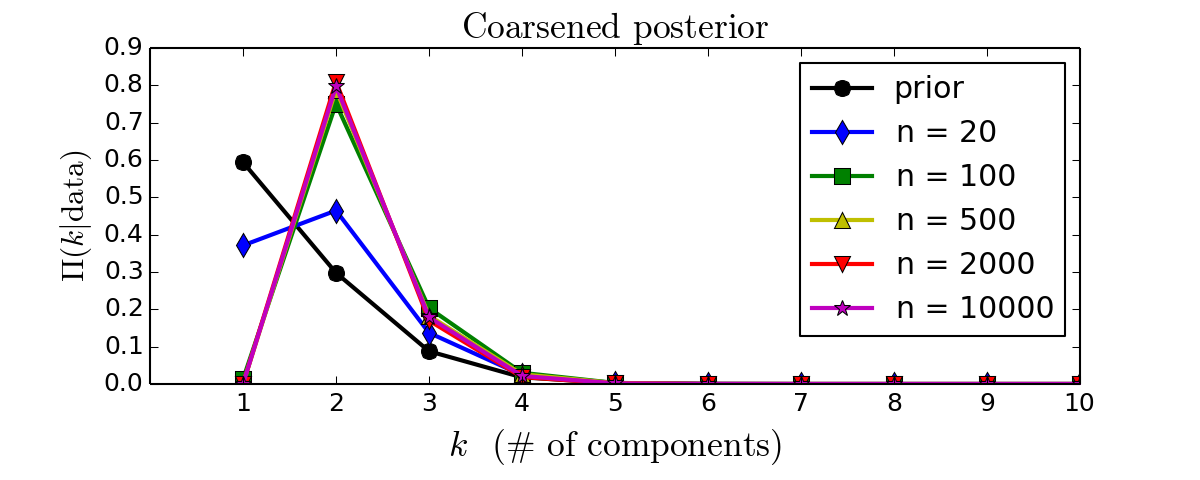}
  \includegraphics[trim=.5cm 0 1.75cm 0, clip, width=0.49\textwidth]{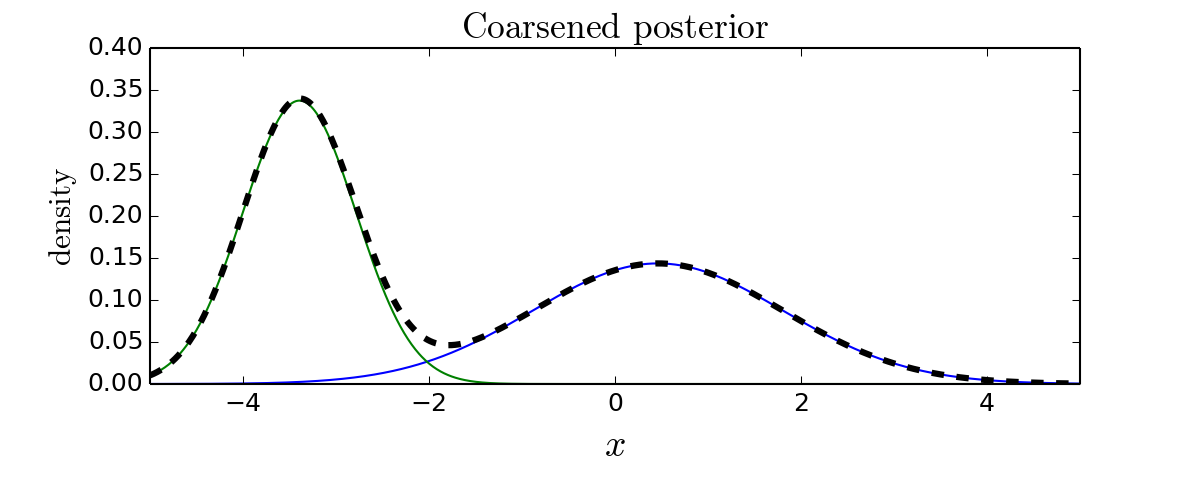}
  \caption{Gaussian mixture with a prior on the number of components $k$, applied to data from a two-component skew-normal mixture.
  Top: Density of the data distribution (blue line), and a histogram of $n=100$ samples.
  Middle left: The posterior on $k$ favors larger and larger values as $n$ increases.
  Bottom left: The coarsened posterior on $k$ stabilizes as $n$ increases, favoring the true number of components, $k = 2$.
  Middle right: Mixture density (dotted black line) and components (solid colors)  for a typical sample from the posterior when $n=10^4$.
  Bottom right: Same for the coarsened posterior.
  See Section~\ref{section:mixtures} for details.}
  \label{figure:skewmix}
\end{figure}

For example, suppose one is using a Gaussian mixture model with a prior on the number of components,
but the data come from a mixture in which the components are not exactly Gaussian. 
In order to fit the data, the posterior will introduce more and more components as the amount of data increases,
and the inferred components will not accurately reflect the true components.
To illustrate, suppose the data are from a two-component mixture of skew-normal distributions; see Figure~\ref{figure:skewmix}.
As shown in the figure, the posterior on the number of components favors larger and larger values
as the sample size increases, and a typical sample from the posterior consists of many small components,
obscuring the two large groups corresponding to the true components.
Meanwhile, using the technique introduced in this paper, one can construct a ``coarsened'' posterior for the same model,
under which the number of components does not continue to grow---see Figure~\ref{figure:skewmix}---and as a result, 
one obtains a more macroscopic interpretation of the data as coming from two large groups.
See Section~\ref{section:mixtures} for details.

Ideally, one would model all aspects of the data generating process completely correctly,
however, this is often impractical for a number of reasons.
First, it may be unrealistic to expect to have sufficient insight into the data generating process to
even write down an adequate model.  Further, this can significantly increase the time and effort
required to design the model, devise and implement reasonably efficient
inference algorithms with reliable performance,
and perhaps develop theoretical guarantees.
Even after all such efforts, one may end up with a complex hierarchical model that is unlikely to be
used in scientific applications, because scientists prefer statistical methods that have a clear and
simple interpretation, and are not comfortable with drawing inferences from complicated models in a
``black box'' fashion.

In fact, in some cases, a simple model may actually be more appropriate than a more complex one, even
when it does not exactly fit the observed data.  For instance, when the underlying phenomenon of
interest is well-described by a simple model, but there is a lack of fit due to contaminating artifacts,
then choosing a more complex model can be viewed as a form of overfitting, and may damage the
resulting inferences.  Meanwhile, even when the underlying phenomenon is not well-described by a simple model,
many times the purpose of a model is to provide a lens through
which to understand the data, rather than just fitting it---such as
when the model is being used as a tool for exploratory analysis---in which case it is essential 
to use interpretable models so that the parameters and latent variables provide insight
into the questions of interest.

There have been advances in robustness to misspecification, with methods such as 
Gibbs posteriors \citep{jiang2008gibbs}, 
disparity-based posteriors \citep{hooker2014bayesian}, 
partial posteriors \citep{doksum1990consistent},
nonparametric approaches \citep{rodriguez2014univariate},
neighborhood methods \citep{liu2009building},
and learning rate adjustment \citep{grunwald2014inconsistency};
see Section~\ref{section:connections} for a discussion of previous work.
However, despite growing recognition of the issue and efforts toward a solution, existing methods tend
to be either limited in scope, computationally prohibitive, or lacking a clear justification.

These considerations lead to the following questions.
Is it possible to draw coherent inferences from a model that may be slightly misspecified?
Can this be done in a computationally-tractable way?
In the context of model averaging and nonparametrics, is there a principled way to be tolerant of
models that are not exactly right, but are close enough in some sense?

In this article, we explore a novel approach to robust Bayesian inference that may provide affirmative
answers to these questions.  Instead of using the standard posterior obtained by conditioning on the event
that the observed data are generated by sampling from the model---which is clearly incorrect when the
model is misspecified---the approach we consider is, roughly speaking, to condition on the event that
the empirical distribution of the observed data is close to the empirical distribution of data sampled
from the model, with respect to some statistical distance on probability measures.  
We refer to this as a coarsened posterior, or c-posterior, for short; 
see Section~\ref{section:method} for a detailed description.

The c-posterior approach has a number of appealing features.
It has a compelling justification---it is valid Bayesian inference based on limited information.
The interpretation is conceptually clear---one does inference with the same model, but
conditioned on a different event than usual.
The c-posterior inherits the continuity properties of the chosen statistical distance, and thus,
automatically exhibits robustness to small departures from the model---that is, small changes to the data distribution
result in small changes to the c-posterior.
Asymptotically, the c-posterior takes a relatively simple form, facilitating computation and analysis.

A particularly attractive case occurs when using neighborhoods based on relative entropy estimates,
since then it turns out that the c-posterior can be approximated by
simply raising the likelihood to a certain fractional power; see Section~\ref{section:relative-entropy}.
Consequently, in this case one can often do approximate inference using standard algorithms,
with no additional computational burden---in fact, the mixing time of Markov chain Monte Carlo (MCMC)
samplers will typically be improved, since the likelihood is tempered.
Further, when using exponential families and conjugate priors, one can even obtain 
analytical expressions for quantities such as a ``robustified'' marginal likelihood.

The main disadvantage of c-posteriors is that sometimes they are less concentrated than
one would like---for instance, if it turns out that the amount of misspecification is less than expected.

An unexpected side benefit of our investigation of c-posteriors is that it reveals an interesting
connection between Gibbs posteriors and approximate Bayesian computation (ABC), two areas of current
research in Bayesian statistics. Roughly, a Gibbs posterior can be thought of as an asymptotic
approximation to a particular ABC posterior; see Section~\ref{section:method}.

The paper is organized as follows.  
First, to demonstrate the basic idea, in Section~\ref{section:toy}
we consider a c-posterior for the simplest possible toy example: Bernoulli trials.
Then, in Section~\ref{section:method}, we describe the c-posterior approach more generally.
We discuss connections with previous work in Section~\ref{section:connections}, and
in Section~\ref{section:theory} we establish various theoretical properties of c-posteriors
regarding asymptotics and robustness. 
In Section~\ref{section:extensions}, we show how the method extends to time series and regression, and in
Section~\ref{section:applications}, we apply the c-posterior approach to 
autoregressive models of unknown order, variable selection in linear regression, 
and mixture models with an unknown number of components.
We close with a brief discussion of possible directions for future work.

\section{Toy example: Bernoulli trials}
\label{section:toy}

For expository purposes, we first introduce the c-posterior in a toy example.  Suppose $X_1,\ldots,X_n$ i.i.d.\ $\sim\Bernoulli(\theta)$
represent the outcomes of $n$ replicates of a laboratory experiment,
and the team of experimenters is interested in testing $\HH_0:\theta=1/2$ versus $\HH_1:\theta\neq 1/2$.  

The standard Bayesian approach is to define a prior probability for each hypothesis, say, $\Pi(\HH_0)=\Pi(\HH_1)=1/2$,
and define a prior density for $\theta$ in the case of $\HH_1$, say, $\theta|\HH_1 \sim\Uniform(0,1)$.
Inference then proceeds based on the posterior probabilities of the hypotheses, $\Pi(\HH_0|x_{1:n})$ and
$\Pi(\HH_1|x_{1:n}) = 1 - \Pi(\HH_0|x_{1:n})$, where $x_{1:n} = (x_1,\ldots,x_n)$.
If the observed data $x_1,\ldots,x_n$ are sampled i.i.d.\ from $\Bernoulli(\theta)$,
then the posterior is guaranteed to converge to the correct
answer, that is, $\Pi(\HH_0 | x_{1:n}) \xrightarrow[]{\mathrm{a.s.}} \I(\theta = 1/2)$ as $n\to\infty$.
(We use $\I(\cdot)$ to denote the indicator function: $\I(E)=1$ if $E$ is true, and $\I(E)=0$ otherwise.)

In reality, however, it is likely that the observed data do not exactly follow the assumed model. 
For instance, some of the experiments may have been conducted under slightly different conditions than others
(such as at different times or by different researchers), or 
some of the outcomes may be corrupted due to human error in carrying out the experiment.
Of course, in such a simple setting as Bernoulli trials, it would be easy to improve the model to account for issues such as these.
However, for more complex models it is often not so easy, as discussed in the introduction,
and we seek a method that works well even with complex models.

Suppose it is known that any such corruption affects the distribution of the data by only a small amount.
We can formulate this mathematically by considering $X_{1:n}$ to represent some hypothetical ``true'' data which do follow the model,
and supposing that the observations $x_{1:n}$ are close to the true data in some distributional sense, but not necessarily equal to it.
A natural way to define distributional ``closeness'' is in terms of the relative entropy 
$D(\hat p_x\|\hat p_X) = \sum_{i = 0}^1 \hat p_x(i)\log (\hat p_x(i)/\hat p_X(i))$ between
the empirical distributions of $x_{1:n}$ and $X_{1:n}$, i.e., $\hat p_x(1) = \bar x$ in this example.

Due to the corruption, it is inappropriate to condition on the true data $X_{1:n}$ being exactly equal to the observed data $x_{1:n}$.
Instead, if it is known that $X_{1:n}$ is close to $x_{1:n}$ in the sense that $D(\hat p_x\|\hat p_X)<r$,
and nothing more is known about the nature of the corruption, then a natural Bayesian approach would be to condition on
the event that $D(\hat p_x\|\hat p_X) < r$, that is, to use  $\Pi\big(\HH_0\,\big\vert\, D(\hat p_x||\hat p_X) < r\big)$.
In other words, rather than conditioning on the data exactly,
condition on a relative entropy neighborhood of the empirical distribution of the data.

In practice, one will typically only have a rough idea about the amount of corruption, 
and thus, it makes sense to put a prior on $r$, say, $R\sim\Exponential(\alpha)$.
This leads us to consider the following ``coarsened'' posterior, or c-posterior, for inferences about $\HH_0$ and $\HH_1$:
\begin{align}\label{equation:toy-coarse}
\Pi\big(\HH_0\,\big\vert\, D(\hat p_x||\hat p_X) < R).
\end{align}
In other words, we consider $\Pi(\HH_0 | Z=1)$ where $Z = \I(D(\hat p_x||\hat p_X) < R)$.
How should we choose $\alpha$? In this example, we can interpret the neighborhood size $r$ in terms of intuitive Euclidean notions
by using the chi-squared approximation to relative entropy,
$D(p\|q)\approx \frac{1}{2}\chi^2(p,q)$ (see Prop.\ \ref{proposition:chi-square-relative-entropy}).
In particular, when $\bar X \approx 1/2$
we have $D(\hat p_x||\hat p_X)\approx 2|\bar x - \bar X|^2$, and thus,
if we expect the corruption to shift the sample mean by no more than $\epsilon$ or so when $\HH_0:\theta = 1/2$ is true,
then it makes sense to choose $\alpha$ so that $\E R \approx 2\epsilon^2$.
Since $\E R = 1/\alpha$ this suggests using $\alpha = 1/(2\epsilon^2)$.
\begin{figure}
  \centering
  \includegraphics[width=0.8\textwidth]{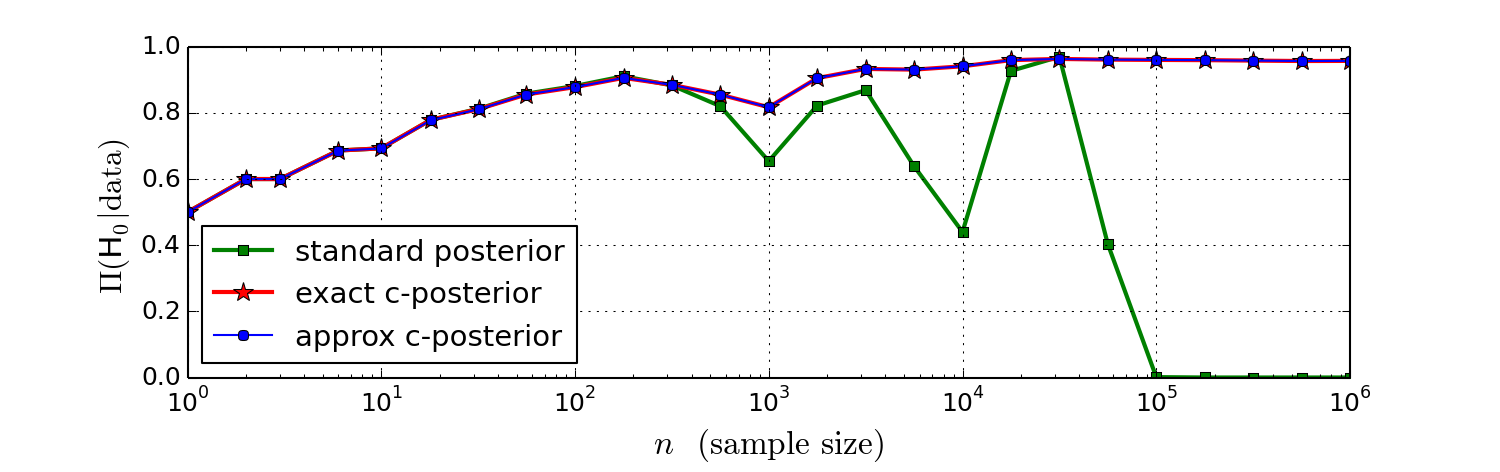}
  \includegraphics[width=0.8\textwidth]{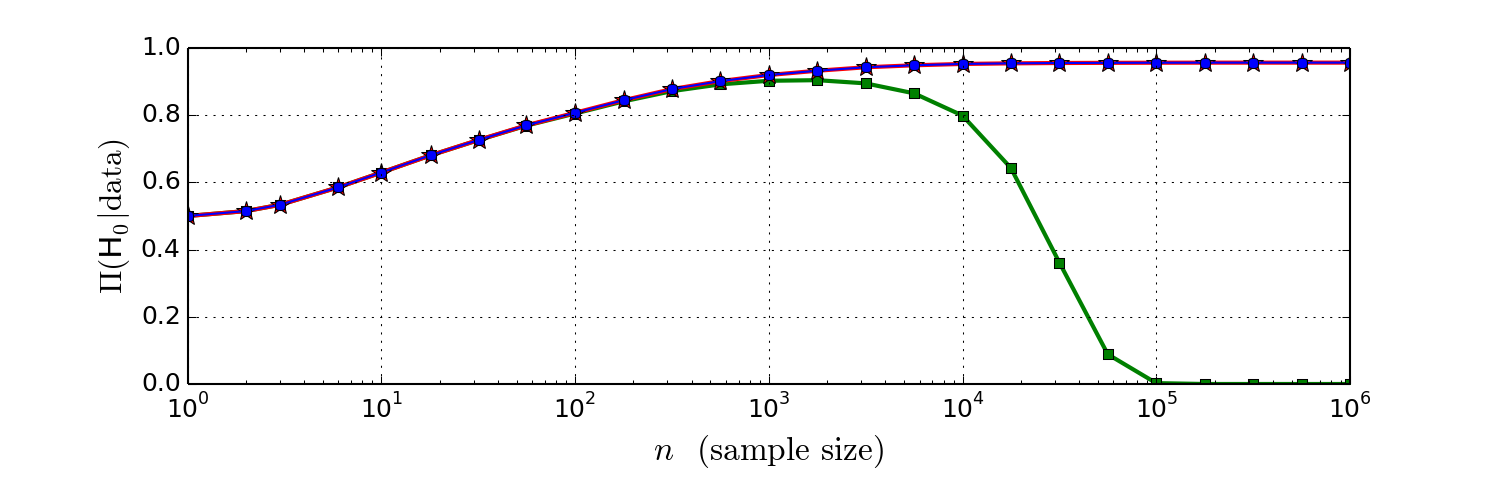}
  \includegraphics[width=0.8\textwidth]{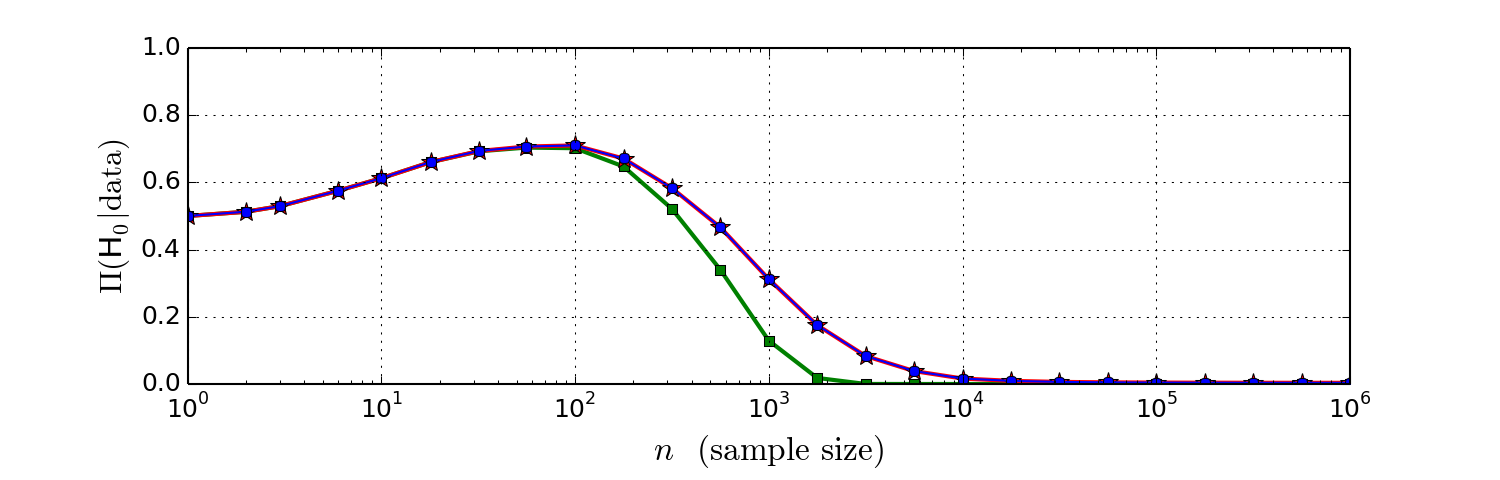}
  \caption{Bernoulli trials example. Top: Results from a single sequence $x_1,x_2,\ldots\iid\Bernoulli(0.51)$. 
  Middle: Average over 1000 sequences $x_1,x_2,\ldots\iid\Bernoulli(0.51)$.
  Bottom: Same as middle, but with 0.56 instead of 0.51. 
  In all three plots, the approximate c-posterior is indistinguishable from the exact c-posterior.}
  \label{figure:Bernoulli}
\end{figure}

In this toy example, the c-posterior in Equation \ref{equation:toy-coarse} can be computed exactly
(see Section~\ref{section:toy-details} for details),
however, in more complex cases, an approximation will be needed. In Section~\ref{section:relative-entropy}, we develop a general approximation
which, when applied to this example, yields
\begin{align}\label{equation:toy-approx}
\Pi\big(\HH_0\,\big\vert\, D(\hat p_x||\hat p_X) < R) \approx 1/\big(1 + 2^{\alpha_n} B(1+\alpha_n \bar x,\,1+\alpha_n(1-\bar x))\big)
\end{align}
where $\alpha_n = 1/(1/n +1/\alpha)$ and $B(a,b)$ is the beta function (see \ref{section:toy-details} for details).
Comparing this to the standard posterior,
\begin{align}\label{equation:toy-standard}
\Pi\big(\HH_0\,\big\vert\, X_{1:n}=x_{1:n}) = 1/\big(1 + 2^n B(1 + n \bar x,\,1 + n(1-\bar x))\big),
\end{align}
note that the only difference is that $n$ has been replaced by $\alpha_n$ in the c-posterior.

To illustrate numerically, suppose we would like to be robust to perturbations affecting $\bar x$ by roughly 
$\epsilon = 0.02$ when $\HH_0$ is true. As described above, this corresponds to $\alpha = 1/(2\cdot 0.02^2)=1250$.
Now, suppose that in reality $\HH_0$ is indeed true, and the data are corrupted in such a way that $x_1,\ldots,x_n$ 
behave like i.i.d.\ samples from $\Bernoulli(0.51)$. 
Figure~\ref{figure:Bernoulli} (top and middle) shows the probability of $\HH_0$ under the standard posterior, the exact c-posterior, and 
the approximate c-posterior (Equations \ref{equation:toy-standard}, \ref{equation:toy-coarse}, and \ref{equation:toy-approx}, respectively),
for increasing values of the sample size $n$. 

When $n$ is small, there is not enough power to distinguish between 0.5 and 0.51, so the standard posterior favors
$\HH_0$ at first (due to the Bartlett--Lindley effect), but as $n$ increases, eventually the posterior probability of $\HH_0$ goes to 0.
(So, when $n$ is large, the standard posterior is not robust to this perturbation.)
Meanwhile, the c-posterior behaves the same way as the standard posterior when $n$ is small,
but as $n$ increases, the c-posterior probability of $\HH_0$ remains high, as desired---thus, the c-posterior remains robust for large $n$.
Note, further, that the curve for the approximate c-posterior is directly on top of the curve for the exact c-posterior---the approximation is so close that the two are indistinguishable.

What if the departure from $\HH_0$ is significantly larger than our chosen tolerance of $\epsilon = 0.02$? 
Does the c-posterior more strongly favor $\HH_1$ in such cases, as it should?
Indeed, it does. Figure~\ref{figure:Bernoulli} (bottom) shows the three posteriors on data $x_1,\ldots,x_n\iid\Bernoulli(0.56)$.
We see that in this case, the c-posterior behaves more like the standard posterior, favoring $\HH_1$ when $n$ is sufficiently large.

It is important to note that, unlike the standard posterior, the c-posterior does not concentrate as $n\to\infty$.
This is appropriate, since in the presence of corruption,
some uncertainty always remains about the true distribution, no matter how much data is observed.

\section{Method}
\label{section:method}



In this section, we describe the c-posterior approach more generally. 
For discussion of connections with previous work, see Section~\ref{section:connections};
in particular, our ideas have been influenced by \citet{lindsay2009model} and \citet{wilkinson2013approximate}.
For the time being, we assume i.i.d.\ data,
but the approach generalizes, for example, to time-series (Section~\ref{section:time-series}) and
regression (Section~\ref{section:regression}).


Suppose we have a model $\{P_\theta:\theta\in\Theta\}$ along with a prior $\Pi$ on
$\Theta$, and suppose there is a point $\theta_I\in\Theta$ representing the parameters
of the \textit{idealized distribution} of the data.
The interpretation here is that $\theta_I$ is the ``true'' state of nature about which one is interested
in making inferences; it may represent some actual underlying truth or may merely be a useful fiction.  
(For context, note that in many scientific endeavors, one employs idealized models that capture the most
important features of the phenomena of interest, without harboring any illusions that they 
completely describe every detail.)
Now, suppose there are some unobserved \textit{idealized data}
$X_1,\dotsc,X_n\in\X$ which are i.i.d.\ from $P_{\theta_I}$,
however, the \textit{observed data} $x_1,\dotsc,x_n\in\X$ are actually a slightly corrupted version of
$X_1,\dotsc,X_n$ in the sense that $d(\PX,\Px) < r$ for some statistical distance $d(\cdot,\cdot)$ 
and some $r > 0$, where $\hat P_{x_{1:n}}=\frac{1}{n}\sum_{i = 1}^n \delta_{x_i}$
denotes the empirical distribution of $x_{1:n} =(x_1,\dotsc,x_n)$.
Suppose $x_1,\ldots,x_n$ behave like i.i.d.\ samples from some $P_o$,
and note that due to the corruption, we expect that $P_o\neq P_{\theta_I}$.
For intuition, consider the diagram in Figure~\ref{figure:diagram}.

\begin{figure}
  \centering
  \includegraphics[width=0.7\textwidth]{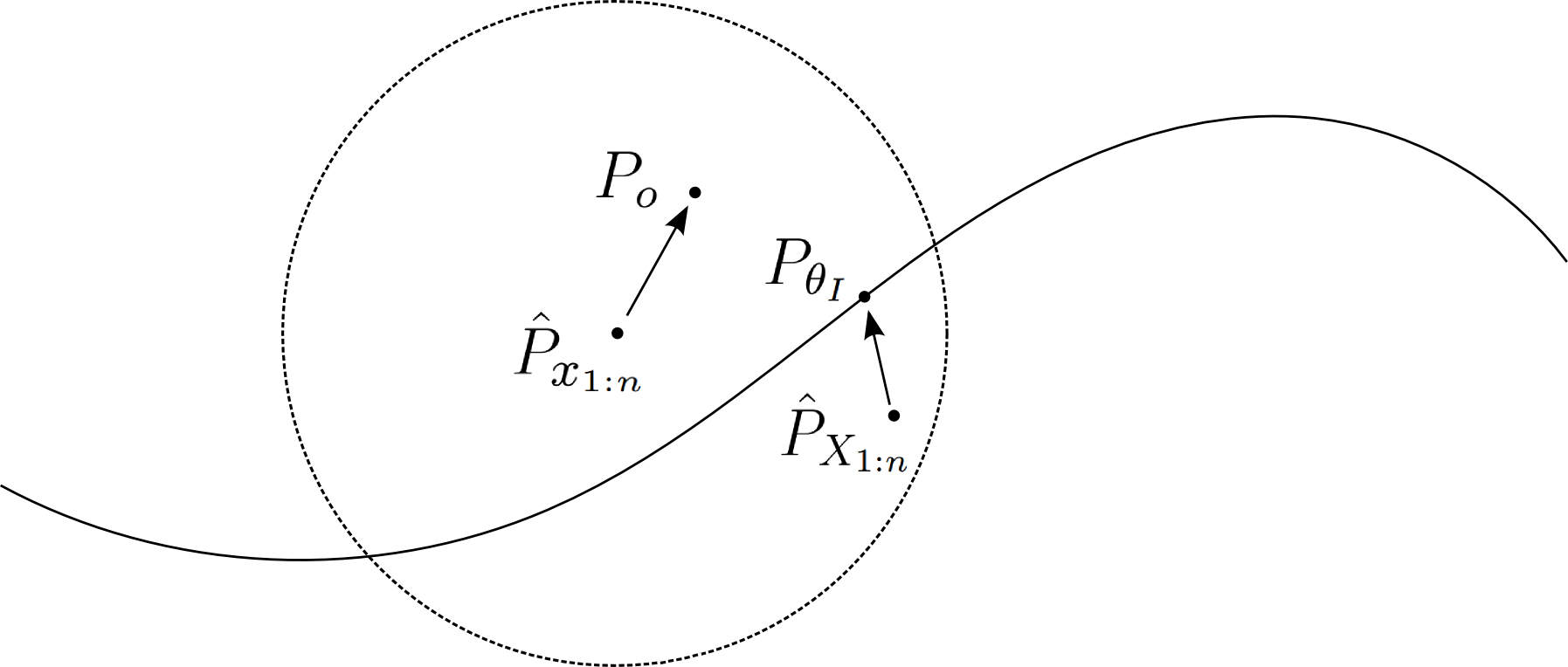}
  \caption{
  Notional schematic diagram of the idea behind the c-posterior. The ambient space is the set of probability distributions on $\X$,
  and the curve represents the subset of distributions in the parametrized family $\{P_\theta:\theta\in\Theta\}$.
  The idealized distribution $P_{\theta_I}$ is a point in this subset, and the empirical distribution $\PX$ of the true data
  converges to $P_{\theta_I}$ as $n\to\infty$.  Although $\PX$ is not observed, it is known to be within an $r$-neighborhood
  of the empirical distribution $\Px$ of the observed data, which, in turn, converges to $P_o$. 
  The basic idea of the c-posterior approach is to condition on the event that $\PX$ is within this neighborhood.
  }
  \label{figure:diagram}
\end{figure}




If there were no corruption (i.e., contamination/misspecification), then 
we should use the standard posterior---that is, we should condition on the event that 
$X_{1:n} = x_{1:n}$---however, due to the corruption this would clearly be incorrect.  
Of course, if one could easily model the corrupting process by which 
$x_{1:n}$ is generated from $X_{1:n}$, then the most sensible approach would be to 
simply incorporate it into the model, but this may be impractical, 
as discussed at length in the introduction.

An alternative approach is to condition on what is known---that is, 
to condition on the event that $d(\PX,\Px) < r$.
In other words, rather than the standard posterior $\Pi(d\theta \mid X_{1:n} = x_{1:n})$, consider 
$\Pi\big(d\theta \mid d(\PX,\Px) < r\big)$.
Since usually one will not have sufficient \textit{a priori} knowledge to choose $r$,
it makes sense to put a prior on it, say $R\sim H$, independently of $\theta$ and $X_{1:n}$. 
Generalizing further, take a sequence of functions $d_n$ such that 
$d_n(X_{1:n},x_{1:n})\geq 0$ is some measure of the discrepancy between $X_{1:n}$ and $x_{1:n}$.  
\begin{definition}
We refer to $\Pi\big(d\theta \mid d_n(X_{1:n},x_{1:n}) < R\big)$ as a \textit{c-posterior}.
\end{definition}
It is useful to note that one can write the c-posterior as
\begin{align}
\Pi\big(d\theta\mid d_n(X_{1:n},x_{1:n}) < R\big)
&\propto \Pi(d\theta)\,\Pr\big(d_n(X_{1:n},x_{1:n}) < R \mid \theta\big)  \notag\\
&= \Pi(d\theta)\int_{\X^n} G(d_n(x'_{1:n},x_{1:n})) P_\theta^n(d x'_{1:n})
\label{equation:exact-c-posterior}
\end{align}
where $G(r) = \Pr(R>r)$.
The intuitive interpretation is that, to use a rough analogy,
this integral is like a convolution of $P_\theta^n$
(the distribution of $X_{1:n}$) with the ``kernel'' $G(d_n(X_{1:n},x_{1:n}))$.
Some readers will recognize the form of the c-posterior in connection with
\textit{approximate Bayesian computation} (ABC),
where it arises due to the nature of the approximation \citep{marjoram2003markov}.
Needless to say, our motivation for using it is completely different than in ABC;
see Section~\ref{section:connections} for discussion.


While one could do inference for the c-posterior using algorithms similar to those used for ABC,
this would be very slow, and would not exploit the fact that the likelihood is tractable.
Instead, we develop approximations to the c-posterior that facilitate efficient inference.
The crudest approach is based on the asymptotics of the c-posterior (see Section~\ref{section:asymptotic}),
which show that for large $n$, 
\begin{align} \label{equation:asymptotic-method}
\Pi\big(d\theta \mid d(\PX,\Px) < R\big)
\approxprop G(d(P_\theta,\Px))\Pi(d\theta)
\end{align}
under mild regularity conditions;
the intuition here is that when $n$ is large, $\PX\approx P_\theta$ with high probability.
Thus, if $d(P_\theta,\Px)$ can be easily computed, 
Equation \ref{equation:asymptotic-method} can be used to perform approximate inference for the c-posterior by,
for instance, using it as a target distribution in Metropolis--Hastings MCMC.
If $R\sim\Exponential(\alpha)$, so that $G(r) = e^{-\alpha r}$,
Equation \ref{equation:asymptotic-method} becomes $\exp(-\alpha d(P_\theta,\Px))\Pi(d\theta)$, 
which some readers will recognize as a Gibbs posterior \citep{jiang2008gibbs}; see Section~\ref{section:connections}.
More generally, similar approximations can be made when using a c-posterior based on $d_n(X_{1:n},x_{1:n})$.

A major disadvantage of an asymptotic approximation as in Equation \ref{equation:asymptotic-method}, however,
is that it is only good when $n$ is sufficiently large; in a sense, it ignores the randomness in $X_{1:n}$.
For the case of relative entropy, we develop a much better approximation that is also applicable for smaller $n$;
see Section~\ref{section:relative-entropy} below.

There are many possible choices of statistical distance $d(\cdot,\cdot)$, 
and the robustness properties of the c-posterior depend on this choice;
in Section~\ref{section:robustness}, we show that as one would expect, the c-posterior 
is robust to changes in $P_o$ that are small with respect to $d(\cdot,\cdot)$.
We use the term \textit{statistical distance} very broadly, to mean any nonnegative function
for assessing discrepancy that is meaningful for a given application.
A few potential candidates would be Kolmogorov--Smirnov (in the univariate setting),
Wasserstein, or a maximum mean discrepancy \citep{gretton2006kernel}.
When $P_\theta$ and $P_o$ admit density functions,
it is also possible to accomodate distances on densities, such as 
relative entropy, Hellinger distance, and various divergences---even though
they may be undefined for empirical distributions---by choosing
$d_n(X_{1:n},x_{1:n})$ to be a consistent estimator of $d(P_\theta,P_o)$.

In the applications presented in this paper (see Section~\ref{section:applications}),
we focus on relative entropy (and variations thereof)
as our choice of $d(\cdot,\cdot)$, since it works out exceptionally nicely in several respects.
In particular, it turns out that in this case there is a trick that makes
it unnecessary to explicitly compute $d_n(X_{1:n},x_{1:n})$.
We discuss this next. 



\subsection{Relative entropy c-posteriors}
\label{section:relative-entropy}

In the case of relative entropy, 
there is an approximation to the c-posterior that improves upon Equation \ref{equation:asymptotic-method}
in two respects: it is applicable for small $n$, and it is extremely easy to work with---in fact,
it simply amounts to tempering the likelihood.
As above, we assume i.i.d.\ data for the moment, and refer to
Sections \ref{section:time-series} and \ref{section:regression} for generalizations.

Suppose $P_o$ and $P_\theta$ (for all $\theta\in\Theta$) have densities $p_o$ and $p_\theta$, respectively,
with respect to some sigma-finite measure $\lambda$. Define
$$ d(P_\theta,P_o) = D(p_o\|p_\theta) = \int p_o(x)\Big(\log\frac{p_o(x)}{p_\theta(x)}\Big)\lambda(d x), $$
and suppose $d_n(X_{1:n},x_{1:n})$ is a consistent estimator of $D(p_o\|p_\theta)$.
If $R\sim\Exponential(\alpha)$, then asymptotically, the c-posterior based on $d_n(X_{1:n},x_{1:n})$ is proportional to
\begin{align}
\label{equation:asyptotic-relative-entropy-posterior}
\exp(-\alpha D(p_o\|p_\theta))\Pi(d\theta) 
& \propto \exp(\alpha {\textstyle\int} p_o \log p_\theta)\Pi(d\theta) \\
& \approx \exp\Big(\alpha \frac{1}{n}\sum_{i=1}^n \log p_\theta(x_i)\Big)\Pi(d\theta) \notag\\
& = \Pi(d\theta) \prod_{i=1}^n p_\theta(x_i)^{\alpha/n} \notag
\end{align}
under mild regularity conditions; see Section~\ref{section:relative-entropy-asymptotic}.
When $n$ is small relative to $\alpha$, however, this is unsuitable as an approximation to
the c-posterior, since in particular,
if $\alpha/n > 1$ then this makes the likelihood more concentrated, rather than less.


Instead, we propose a better alternative, based on a central limit theorem approximation.
The derivation is well-founded when the sample space $\X$ has finitely-many elements,
and the extension to general $\X$, while heuristic, is intuitively sensible.
When $|\X|<\infty$, a natural choice of 
$d_n(X_{1:n},x_{1:n})$ is simply $D(\hat p_{x_{1:n}}\|\hat p_{X_{1:n}})$, that is, the relative
entropy of the empirical densities. Assume $R\sim\Exponential(\alpha)$.
Then by Equation \ref{equation:exact-c-posterior}
and an approximation detailed in Section~\ref{secton:small-sample-justification},
\begin{align*}
\Pi\big(d\theta\mid d_n(X_{1:n},x_{1:n}) < R\big)
&\propto \E\big(\exp(-\alpha D(\hat p_{x_{1:n}}\|\hat p_{X_{1:n}}))\mid \theta\big)\Pi(d\theta) \\
&\approx (n\zeta_n/\alpha)^{\frac{|\X|-1}{2}} \exp(-n\zeta_n D(\hat p_{x_{1:n}}\| p_\theta))\Pi(d\theta) \\
&\propto \exp\big(\zeta_n \textstyle{\sum_{i=1}^n} \log p_\theta(x_i)\big)\Pi(d\theta)
\end{align*}
where
\begin{align}\label{equation:zeta}
\zeta_n = \frac{1/n}{1/n +1/\alpha} = \frac{1}{1 + n/\alpha}.
\end{align}
The rough idea is that $\hat p_{X_{1:n}}$ is approximately multivariate normal with mean $p_\theta$ and precision of order $n$, and
the exponentiated relative entropy approximates a normal density with mean $\hat p_{X_{1:n}}$ and precision of order $\alpha$,
so the expectation above behaves like the convolution of two normals, and the resulting precision is of order $1/(1/n + 1/\alpha) = n \zeta_n$.
As demonstrated in Figure~\ref{figure:Bernoulli}, this approximation can be quite accurate even when $n$ is small,
although more experimentation is needed to assess its accuracy in more general situations.
Note that $\zeta_n \approx \alpha/n$ when $n \gg \alpha$, and $\zeta_n \approx 1$ when $\alpha \gg n$;
thus, we refer to this as the small-sample correction to the asymptotic approximation.


This leads to the following useful approximation to the relative entropy c-posterior:
\begin{align}
\label{equation:zeta-posterior}
\Pi\big(d\theta\mid d_n(X_{1:n},x_{1:n}) < R\big)
\,\,\approxprop\,\, \Pi(d\theta) \prod_{i=1}^n p_\theta(x_i)^{\zeta_n}.
\end{align}
Note, in particular, that this enables one to approximate the c-posterior
without explicitly computing the relative entropy estimates $d_n(X_{1:n},x_{1:n})$,
which would normally involve computing a density estimate of $p_o$ in order to handle the
$\int p_o \log p_o$ term in $D(p_o\|p_\theta)$.
Since this term is constant with respect to $\theta$, it is absorbed into the constant of proportionality,
allowing one to bypass this density estimation step,
which would be both computationally expensive and statistically inefficient, especially in high dimensions.

\begin{definition}
Given $\zeta\in[0,1]$, 
we refer to $\prod_{i=1}^n p_\theta(x_i)^{\zeta}$ as a \textit{power likelihood},
and to the distribution proportional to $\Pi(d\theta) \prod_{i=1}^n p_\theta(x_i)^{\zeta}$ 
as a \textit{power posterior}.
\end{definition}
There are a number of other methods in which a power likelihood is employed (see Section~\ref{section:connections}),
however, to our knowledge, the form of power we use (i.e., $\zeta = \zeta_n$),
and its theoretical justification, are novel.

A useful interpretation of our power posterior is that
it corresponds to adjusting the sample size from $n$ to $n \zeta_n$.
Thus, roughly speaking, by choosing a particular value of $\alpha$,
one makes the power posterior tolerant of all $\theta$'s for which a sample of size $\alpha$
from $P_o$ could plausibly have come from $P_\theta$. 
This idea is closely related to the model credibility index of \citet{lindsay2009model}.
This suggests the interesting possibility of approximating the c-posterior by
taking random subsets of size $\alpha$, and combining the resulting posteriors
in some way; this might have certain advantages in terms of computation and implementation,
however, we do not explore it in this article.

Due to its simple form, inference using the power posterior is often easy,
or at least, no harder than inference using the ordinary posterior.
We discuss three commonly-occuring cases:
analytical solution in the case of exponential families with conjugate priors,
Gibbs sampling in the case of conditionally-conjugate priors, 
and Metropolis--Hastings MCMC more generally.

\subsubsection{Power posterior with conjugate priors}
\label{section:conjugate-priors}

When using exponential families with conjugate priors, one can often obtain
analytical expressions for integrals with respect to the power posterior.
Suppose $p_\theta(x) = \exp\big(\theta^\T s(x) - \kappa(\theta)\big)$,
where $s(x)=(s_1(x),\ldots,s_k(x))^\T$ are the sufficient statistics,
and suppose $\Pi(d\theta) =\pi_{\xi,\nu}(\theta)d\theta$ where
$\pi_{\xi,\nu}(\theta) = \exp\big(\theta^\T\xi - \nu\kappa(\theta) - \psi(\xi,\nu)\big)$,
noting that this defines a conjugate family.
Then the power posterior is proportional to
\begin{align}\label{equation:power-posterior}
   \pi_{\xi,\nu}(\theta)\prod_{i = 1}^n p_\theta(x_i)^{\zeta_n}
\propto \exp\Big(\theta^\T \big(\xi + \zeta_n\textstyle{\sum_i} s(x_i)\big) - (\nu + n\zeta_n)\kappa(\theta)\Big)
\propto \pi_{\xi_n,\nu_n}(\theta),
\end{align}
where $\xi_n = \xi + \zeta_n\sum_i s(x_i)$ and $\nu_n =\nu + n\zeta_n$,
and thus, the power posterior remains in the conjugate family.

For most conjugate families used in practice, simple analytical expressions are available for
the log-normalization constant $\psi(\xi,\nu)$
as well as for many integrals with respect to $\pi_{\xi,\nu}(\theta)$.
This enables one to obtain analytical expressions for many quantities of inferential interest
under the power posterior, thus providing approximations to the corresponding quantities under the
relative entropy c-posterior.
For instance, one obtains a marginal power likelihood,
$$ \int_\Theta \pi_{\xi,\nu}(\theta)\prod_{i = 1}^n p_\theta(x_i)^{\zeta_n} d\theta
= \exp\big(\psi(\xi_n,\nu_n) - \psi(\xi,\nu)\big),$$
which can be used to compute robustified Bayes factors or a 
robust posterior on models, in the context of model inference.
This is robust to small perturbations to $P_o$ (in the sense of relative entropy),
whereas the usual model inference/selection procedures can be very sensitive to such perturbations, for large $n$;
see Section~\ref{section:robustness} for details.

In Section~\ref{section:toy}, we used this approach in the toy example involving Bernoulli trials.
In Section~\ref{section:autoregressive}, we apply it to perform robust inference
for the order of an autoregressive model.


\subsubsection{MCMC on the power posterior}

Often, it is desirable to place conditionally-conjugate priors on the parameters of an
exponential family---for instance, placing independent normal and inverse-Wishart
priors on the mean and covariance of a normal distribution. In such cases, 
one can use Gibbs sampling on the power posterior, because for each parameter given the others,
we are back in the case of a conjugate prior, and thus---as shown by
Equation \ref{equation:power-posterior}---the full conditionals belong to the conjugate family,
making them easy to sample from.
In Section~\ref{section:variable-selection}, we use Gibbs sampling for variable selection in linear regression 
with the power posterior.

More generally, samples can be drawn from the power posterior by using Metropolis--Hastings MCMC,
with the power likelihood in place of the usual likelihood.
In Section~\ref{section:mixtures}, we use Metropolis--Hastings for inference in mixtures with a prior on the number of components,
with the power posterior.

By a stroke of luck, the mixing time for MCMC with the power posterior 
will often be better than with the standard posterior, since
raising the likelihood to a fractional power (i.e., between 0 and 1) has the effect of flattening it,
enabling the sampler to more easily move through the space, particularly when
there are multiple modes and $n$ is large. Indeed, raising the likelihood to a 
fractional power---also known as tempering---is sometimes done in 
more complex MCMC schemes in order to improve mixing time, a well-known example
being MC\textsuperscript{3} \citep{geyer1991markov}.

Thus, generally speaking, it is a straightforward matter to use MCMC
for sampling from the power posterior. However, there is a subtle point
that should be carefully noted. Often, latent variables are introduced 
into an MCMC scheme in order to facilitate moves or to improve mixing, and sometimes,
such latent variables do not work in the same way for the power posterior.
For example, in a mixture model, say, $\sum_{i = 1}^k w_i f_{\varphi_i}(x)$, latent variables $z_1,\ldots,z_n$ indicating which
component each datapoint comes from are often introduced, so that the full conditional distributions for 
Gibbs sampling from $w$, $\varphi$, and $z$ take nice and simple forms.
However, when using the power posterior, the likelihood is $\prod_{j = 1}^n \big(\sum_{i = 1}^k w_i f_{\varphi_i}(x_j)\big)^{\zeta_n}$,
and it seems that introducing $z_1,\ldots,z_n$ no longer leads to nice full conditionals.
On the other hand, it may be possible to use a different set of latent variables; see \citet{walker13bayesian} for the case of mixtures.

\section{Connections with previous work}
\label{section:connections}


The c-posterior is mathematically equivalent to the type of posterior approximation resulting from approximate Bayesian computation (ABC)
\citep{tavare1997inferring,marjoram2003markov,beaumont2002approximate,wilkinson2013approximate}---indeed,
ABC provided part of our inspiration for considering this form of posterior.  However, there are some crucial distinctions to note. 
First, the motivation here is completely different than with ABC: we are concerned with robustness to misspecification,
while ABC is concerned with inference in models with intractable likelihoods.
Generally speaking, we assume the likelihood is easily computed, which makes inference much more computationally efficient.
Another major difference is that in ABC, the coarsened posterior is viewed as an undesirable side effect of the approximate nature of ABC,
while from our perspective, it is precisely the object of interest---in other words, for us it is an asset, not a liability. 


The c-posterior can also be viewed as conditioning on partial information, a technique which is often used to improve robustness
\citep{doksum1990consistent,pettitt1983likelihood,hoff2007extending,dunson2005approximate}; also see \citet{cox1975partial}.
Usually, however, this is done by conditioning on some insufficient statistic; for example, \citet{doksum1990consistent}
perform robust Bayesian inference for a location parameter by conditioning only on the sample median, rather than the whole sample.
Our approach of conditioning on a distributional neighborhood is quite different.


Gibbs posteriors have recently been introduced as a very general framework for updating prior beliefs using 
a generalized ``likelihood'' \citep{jiang2008gibbs,zhang2006information,li2014general,bissiri2013general}.
Under certain conditions, for $n$ sufficiently large, 
the c-posterior is approximately proportional to $\exp(-\alpha d(P_\theta,\Px))\Pi(d\theta)$, 
which can be viewed as a Gibbs posterior with ``risk'' $d(P_\theta,\Px)$.
In research involving Gibbs posteriors, an issue of current interest is how to choose $\alpha$
so that the concentration of the posterior is appropriately calibrated.
The fact that Gibbs posteriors can be interpreted as an approximation to a coherent Bayesian procedure (the c-posterior)
may provide insight into this calibration problem.

A number of researchers have employed a form of power likelihood obtained by raising the likelihood to a power between 0 and 1.
Usually, this is done for reasons completely unrelated to robustness, such as
marginal likelihood approximation \citep{friel2008marginal},
improved MCMC mixing \citep{geyer1991markov},
consistency in nonparametric models \citep{walker2001bayesian,zhang2006eps},
discounting historical data \citep{ibrahim2000power}, or
objective Bayesian model selection \citep{ohagan1995fractional}. 
However, recently, the robustness properties of power likelihoods have started to be noticed:
\citet{grunwald2014inconsistency} provide an in-depth study of a simulation example in which a power posterior
exhibits improved robustness to misspecification, and they propose a method for choosing the power;
also see \citet{grunwald2011safe,grunwald2012safe}.
Nonetheless, in all such previous research, a fixed power is used, rather than one tending to 0 as $n\to\infty$. 
It seems that neither the form of power likelihood we use, nor the theoretical motivation for it, have appeared in any prior work.

Conceptually speaking, the existing methods that are perhaps most similar to the idea of the c-posterior are
goodness-of-fit tests that assess whether the data distribution is within a neighborhood of the model space
\citep{rudas1994new,goutis1998model,dette2003some,liu2009building}, 
however, the methods used previously are very different from ours. 
Closely related to such work is the model credibility index of \citet{lindsay2009model}, a concept which has
heavily influenced our thinking in the development of the c-posterior.

\section{Theory}
\label{section:theory}



    


   

In this section, we establish the asymptotic form of c-posteriors (Section~\ref{section:asymptotic})
and their robustness properties (Section~\ref{section:robustness}).
Let $\X$ and $\Theta$ be standard Borel spaces,
and let $\M$ denote the space of probability measures on $\X$,
equipped with the weak topology. Let $\{P_\theta:\theta\in\Theta\}\subseteq\M$ be a family of probability
measures on $\X$ such that $\theta\mapsto P_\theta(A)$ is measurable for all measurable subsets $A\subseteq\X$.
Let $\Pi$ be a prior measure on $\Theta$, and consider the following model:
\begin{align*}
    & \btheta \sim \Pi, \\
    & X_1,\ldots,X_n|\btheta \iid P_\bbtheta, \mbox{ and} \\
    & R \sim H, \mbox{ independently of } \btheta,X_{1:n},
\end{align*}
where $H$ is a distribution on $[0,\infty)$.  Note that we use (bold) $\btheta$ for the random variable, and $\theta$ for particular values. 
Define 
    $$G(r) = \Pr(R > r).$$
Now, suppose the observed data $x_1,\ldots,x_n\in\X$ behave like i.i.d.\ samples from some $P_o\in\M$. 
Let $d:\M\times\M\to[0,\infty]$, and for $n\in\{1,2,\ldots\}$, let $d_n:\X^n\times\X^n\to[0,\infty]$.
It is assumed that $\theta \mapsto d(P_\theta,P)$ is measurable for all $P\in\M$, and $d_n(\cdot,\cdot)$ is measurable for each $n$.

\subsection{Asymptotic form of the c-posterior}
\label{section:asymptotic}

The c-posterior takes a simple form as $n\to\infty$, under mild regularity conditions.
We prove a general convergence theorem for c-posteriors (Theorem \ref{theorem:asymptotic}) and 
then apply it to c-posteriors derived from relative entropy (Corollary \ref{corollary:relative-entropy-asymptotic})
and weakly-continuous distances (Corollary \ref{corollary:weakly-cts-asymptotic}).


\subsubsection{Convergence theorem}

The following basic lemma captures the underlying principle at work in establishing both the asymptotic form of the c-posterior
(Theorem \ref{theorem:asymptotic}) as well as its robustness (Theorem \ref{theorem:continuity}).

\begin{lemma}
    \label{lemma:limit}
    If $U,U_n,V,W\in\R\cup\{\infty\}$ are random variables such that $U_n \xrightarrow[n\to\infty]{\mathrm{a.s.}} U$,
    $\Pr(U=V)=0$, $\Pr(U<V)>0$, and $\E|W|<\infty$, then $$\E(W\mid U_n<V) \xrightarrow[n\to\infty]{} \E(W\mid U<V).$$
\end{lemma}
All proofs for this section have been placed in Section~\ref{section:theory-proofs}.
The following condition is necessary to avoid certain pathologies; it is always satisfied, for instance, when
$d(P_\bbtheta,P_o) < \infty$ with positive probability
and $R$ has a density with respect to Lebesgue measure that is positive on $[0,\infty)$.

\begin{condition}
\label{condition:R}
Assume $\Pr(d(P_\bbtheta,P_o) = R) = 0$ and $\Pr(d(P_\bbtheta,P_o) < R) > 0$.
\end{condition}

We use $\Rightarrow$ to denote convergence with respect to the weak topology.

\begin{theorem}
\label{theorem:asymptotic}
If $d_n(X_{1:n},x_{1:n})\xrightarrow[]{\mathrm{a.s.}} d(P_\bbtheta,P_o)$ as $n\to\infty$ and Condition \ref{condition:R} is satisfied, then
\begin{align}
    \label{equation:asymptotic-posterior}
    \Pi\big(d\theta\mid d_n(X_{1:n},x_{1:n}) < R\big)
    \xLongrightarrow[n\to\infty]{}
    \Pi\big(d\theta\mid d(P_\bbtheta,P_o) < R\big)
    \propto
    G\big(d(P_\theta,P_o)\big) \Pi(d\theta),
\end{align}
and in fact,
\begin{align}
    \label{equation:asymptotic-posterior-expectation}
    \E\big(h(\btheta)\mid d_n(X_{1:n},x_{1:n}) < R\big)
    \xrightarrow[n\to\infty]{}
    \E\big(h(\btheta)\mid d(P_\bbtheta,P_o) < R\big)
    = \frac{\E h(\btheta) G\big(d(P_\bbtheta,P_o)\big)}{\E G\big(d(P_\bbtheta,P_o)\big)}
\end{align}
for any $h\in L^1(\Pi)$, i.e., any measurable $h:\Theta\to\R$ such that $\int|h(\theta)|\Pi(d\theta) < \infty$.
\end{theorem}

As noted earlier, a case of particular interest arises when $R\sim\Exponential(\alpha)$,
since then $G(r)=e^{-\alpha r}$ and the resulting asymptotic c-posterior is proportional to $\exp(-\alpha d(P_\theta,P_o))\Pi(d\theta)$,
by Theorem \ref{theorem:asymptotic}.
This is asymptotically equivalent to $\exp(-\alpha d(P_\theta,\Px))\Pi(d\theta)$,
provided that $d(P_\bbtheta,\Px)\xrightarrow[]{\mathrm{a.s.}} d(P_\bbtheta,P_o)$,
and as discussed in Section~\ref{section:connections}, this is precisely the form of a Gibbs posterior;
thus, a Gibbs posterior can be interpreted as a large-sample approximation to a c-posterior.
It is also worth noting that if $R = r_0$ with probability $1$ for some $r_0>0$, then $G(r) = \I(r<r_0)$,
and by Theorem \ref{theorem:asymptotic} the asymptotic c-posterior is proportional to $\I(d(P_\theta,P_o)<r_0)\Pi(d\theta)$,
i.e., it is zero outside the $r_0$ neighborhood of $P_o$ and reverts to the prior inside.

\subsubsection{Application to relative entropy}
\label{section:relative-entropy-asymptotic}

Suppose $P_o$ has density $p_o$ and $P_\theta$ has density $p_\theta$ for each $\theta\in\Theta$.

\begin{corollary}
\label{corollary:relative-entropy-asymptotic}
    Suppose $d_n(X_{1:n},x_{1:n})$ is an almost-surely consistent estimator of $D(p_o\|p_\bbtheta)$,
    i.e., $d_n(X_{1:n},x_{1:n}) \overset{\mathrm{a.s.}}{\longrightarrow} D(p_o\|p_\bbtheta)$.
    If $d(P_\theta,P_o) = D(p_o\|p_\theta)$ and Condition \ref{condition:R} is satisfied, then
    Equations \ref{equation:asymptotic-posterior} and \ref{equation:asymptotic-posterior-expectation} hold.
\end{corollary}
This establishes the asymptotic form of the relative entropy c-posterior as claimed in
Equation \ref{equation:asyptotic-relative-entropy-posterior}.


\subsubsection{Application to weakly-continuous distances}
\label{section:weakly-cts-asymptotic}

Recall that $\hat P_{x_{1:n}} = \frac{1}{n}\sum_{i=1}^n \delta_{x_i}$ denotes the empirical distribution of $x_{1:n}$.

\begin{corollary}
\label{corollary:weakly-cts-asymptotic}
    Suppose $d:\M\times\M\to[0,\infty]$ has the property that
    $d(P_n,Q_n)\to d(P,Q)$ whenever $P_n\Rightarrow P$ and $Q_n\Rightarrow Q$.
    If Condition \ref{condition:R} is satisfied, then
    Equations \ref{equation:asymptotic-posterior} and \ref{equation:asymptotic-posterior-expectation} hold
    when $d_n(X_{1:n},x_{1:n}) = d(\hat P_{X_{1:n}},\hat P_{x_{1:n}})$.
\end{corollary}


\subsection{Robustness properties} 
\label{section:robustness}

Here, we show that when $n$ is large, the standard posterior can be strongly affected by small changes
to the observed data distribution $P_o$, particularly when performing model inference (Section~\ref{section:sensitivity-model}),
while c-posteriors are robust to small changes in $P_o$ (Section~\ref{section:continuity}).
To see roughly why the standard posterior is not robust, note that
\begin{align*}
    \Pi(d\theta\mid X_{1:n}=x_{1:n})&\propto \exp\Big(\sum_{i = 1}^n \log p_\theta(x_i)\Big)\Pi(d\theta)
\approx \exp\big(n {\textstyle\int} p_o \log p_\theta\big)\Pi(d\theta) \\
&\propto \exp(-n D(p_o \| p_\theta))\Pi(d\theta),
\end{align*}
assuming the densities $p_\theta$ and $p_o$ exist.
Due to the $n$ in the exponent, even a slight change to $p_o$ can dramatically change the posterior.
On the other hand, by comparison, the relative entropy c-posterior with
$R\sim\Exponential(\alpha)$ is asymptotically proportional to $\exp(-\alpha D(p_o \| p_\theta))\Pi(d\theta)$,
suggesting that the c-posterior should remain stable in the limit as $n\to\infty$;
below we make this precise, for a general choice of $d(\cdot,\cdot)$.


\subsubsection{Lack of robustness of model inference under the standard posterior}
\label{section:sensitivity-model}

Suppose that for each $k$ in some countable index set,
we have a model $\M_k = \{P_\theta : \theta\in\Theta_k\}$, where $\Theta_k$ is a $t_k$-dimensional Euclidean space. 
Let $\pi(k)$ be a prior on the model index $k$, and for each $k$, let $\pi_k$ be a probability
density with respect to Lebesgue measure on $\Theta_k$; this induces a prior $\Pi$
on the disjoint union $\Theta = \bigcup_k \Theta_k$.

It is well-known that, under mild regularity conditions, the marginal likelihood
$p(x_{1:n}|k) = \int_{\Theta_k} p(x_{1:n}|\theta)\pi_k(\theta) d\theta$ has the asymptotic representation
$$ p(x_{1:n}|k)\sim \frac{p(x_{1:n}|\theta_k^n)\pi_k(\theta_k^*)}{|\det H(\theta_k^*; p_o)|^{1/2}}\Big(\frac{2\pi}{n}\Big)^{t_k/2}, $$
as $n\to\infty$, where $\theta_k^n =\argmax_{\theta\in\Theta_k} p(x_{1:n}|\theta)$ is the maximum likelihood estimator for model $k$,
$\theta_k^*=\argmin_{\theta\in\Theta_k} D(p_o\|p_\theta)$ is the minimal Kullback--Leibler (KL) point within model $k$,
and $H(\theta; p_o) = -\int p_o \big(\nabla_\theta^2 \log p_\theta \big)$.
Here, $a_n\sim b_n$ means $a_n/b_n\to 1$.
Letting $f_n(k) = -\frac{1}{n}\log p(x_{1:n}|\theta_k^n)$, this implies that
\begin{align}\label{equation:laplace-type}
p(x_{1:n}|k)\sim c_k e^{- n f_n(k)} n^{- t_k/2}
\end{align}
for a constant $c_k$ not depending on $n$ or $x_{1:n}$. Typically, $f_n(k)\to f(k) := D(p_o\|p_{\theta_k^*}) -\int p_o \log p_o$.
Note that $f(k')<f(k)$ if and only if model $k'$ is closer to $p_o$ than model $k$ in terms of minimal KL divergence;
also, note that the marginal likelihood automatically penalizes more complex models via the $n^{-t_k/2}$ factor.

Given such an asymptotic representation, it is easy to see that
for any $k$, if there exists $k'$ such that $f(k')<f(k)$, then $\pi(k|x_{1:n})\to 0$ as $n\to\infty$.
Consequently, even the slightest change to $p_o$ can result in major shifts in the posterior on $k$, when $n$ is large.
For instance, it often happens that the models are nested, e.g., $\M_1\subseteq\M_2\subseteq\cdots$ and $t_1<t_2<\cdots$.  
This is the case, for example, when $\M_k$ consists of $k$-component mixtures, or $k$th-order autoregressive models;
variable selection is slightly more complicated but ultimately similar.
If the collection of models is correctly specified with respect to $p_o$,
then there is some minimal $k'$ such that $D(p_o\|p_{\theta_{k'}^*})=0$,
and thus $\pi(k|x_{1:n})\to 0$ for all $k<k'$
(and typically, the posterior on $k$ will concentrate at this $k'$).
However, even the slightest perturbation to $p_o$ will usually result in either
(a) an increase in this minimal $k'$, 
or (b) a situation where $\inf_k D(p_o\|p_{\theta_k^*})$ is not attained at any $k$, causing the posterior on $k$ to diverge,
in the sense that $\pi(k|x_{1:n})\to 0$ for all $k$.
Hence, model inference with the standard posterior is not robust.

\subsubsection{Robustness of the c-posterior}
\label{section:continuity}

The definition of robustness, roughly speaking,
is that small changes to the distribution of the data result in small changes to the resulting inferences.
This can be formalized by requiring that asymptotically, the outcome of an inference procedure
be continuous as a function of $P_o$, with respect to some topology (the weak topology being a standard choice) \citep{huber2004robust}.  
From this perspective, the lack of robustness of the standard posterior can be thought of as a lack of continuity with respect to $P_o$,
asymptotically. 

We show in the following theorem that 
the asymptotic c-posterior inherits the continuity properties of whatever distance $d(\cdot,\cdot)$ is used to define it.
Consequently, the c-posterior will be robust to perturbations to $P_o$, provided that $d(\cdot,\cdot)$ is chosen appropriately.
In the terminology of Section~\ref{section:method},
if the observed data distribution $P_o$ is close to the ideal data distribution $P_{\theta_I}$, then the c-posterior 
will be close to what it would be if $P_o = P_{\theta_I}$.

To interpret the theorem, recall that on any metric space, a function $f(x)$ is continuous if and only if
$x_m\to x$ implies $f(x_m) \to f(x)$. 
Thus, to show continuity as a function of $P_o$ (in some topology), one must show that if $P_m \to P_o$,
then the resulting sequence of asymptotic c-posteriors converges as well. In fact, if $d(\cdot,\cdot)$ is continuous (in this same topology), 
then the theorem shows a bit more than that, since then $P_m \to P_o$ implies $d(P_\theta,P_m) \to d(P_\theta,P_o)$. 

\begin{theorem}
\label{theorem:continuity}
If $P_1,P_2,\ldots\in\M$ such that $d(P_\theta,P_m)\xrightarrow[m\to\infty]{} d(P_\theta,P_o)$
for $\Pi$-almost all $\theta\in\Theta$,
and Condition \ref{condition:R} is satisfied, then for any $h\in L^1(\Pi)$,
$$ \E\big(h(\btheta)\mid d(P_\bbtheta,P_m) < R\big) \longrightarrow \E\big(h(\btheta)\mid d(P_\bbtheta,P_o) < R\big)$$
as $m\to\infty$, and in particular,
$\Pi\big(d\theta \mid d(P_\bbtheta,P_m) < R\big) \Longrightarrow \Pi\big(d\theta \mid d(P_\bbtheta,P_o) < R\big)$.
\end{theorem}


In particular, Theorem \ref{theorem:continuity} implies that the c-posterior is robust in the context of model inference, since taking
$h(\theta) = \I(\theta\in\Theta_k)$, we have
$$ \Pi\big(\Theta_k\mid d(P_\bbtheta,P_m) < R\big) \longrightarrow \Pi\big(\Theta_k\mid d(P_\bbtheta,P_o) < R\big)$$
as $m\to\infty$, under the assumptions of the theorem.

\section{Extensions}
\label{section:extensions}


\subsection{Time-series c-posterior based on relative entropy rate}
\label{section:time-series}

Suppose the sequence of observed data $(x_1,\ldots,x_n)$ is a partial sample from a stationary and ergodic process
with distribution $P_o$,
and suppose the model $\{P_\theta:\theta\in\Theta\}$ consists of stationary finite-order Markov processes.
Assume that for some sigma-finite measure $\mu$ on $\X$, for all $n\in\{1,2,\ldots\}$ and all $\theta\in\Theta$,
the finite-dimensional distributions have densities $p_o(x_1,\ldots,x_n)$ and
$p_\theta(x_1,\ldots,x_n)$ with respect to the product measure $\mu^n$,
and assume $\E_{P_o}|\log p_o(X_{1:n})|<\infty$ and $\E_{P_o}|\log p_\theta(X_{1:n})|<\infty$.

A natural way of assessing the discrepancy between the processes $P_o$ and $P_\theta$
is by the relative entropy rate \citep{gray1990entropy}, 
$$ \D(P_o\|P_\theta) = \lim_{n\to\infty} \frac{1}{n} D\big(p_o(x_{1:n})\|p_\theta(x_{1:n})\big). $$
Suppose $d_n(X_{1:n},x_{1:n})$ is an a.s.-consistent estimator of $\D(P_o\|P_\theta)$
when $(X_1,X_2,\ldots)\sim P_\theta$ and $(x_1,x_2,\ldots)\sim P_o$, and consider the c-posterior
$\Pi\big(d\theta\mid d_n(X_{1:n},x_{1:n})<R\big)$, with $R\sim\Exponential(\alpha)$. 
Then by Lemma \ref{lemma:limit}, the asymptotic c-posterior is 
$$ \Pi(d\theta\mid \D(P_o\| P_\theta)<R)\propto \exp(-\alpha \D(P_o\|P_\theta))\Pi(d\theta). $$
If $P_\theta$ is $k$th-order Markov, then
\begin{align*}
\D(P_o\|P_\theta) = -\Hcal(P_o) - \E_{P_o}\log p_\theta(X_{k +1}|X_1,\ldots,X_k)
\end{align*}
where $\Hcal(P_o)$ is the entropy rate of $P_o$, which we assume is finite \citep[][Lemma 2.4.3]{gray1990entropy}.
Further, when $(x_1,x_2,\ldots)\sim P_o$,
$$ \frac{1}{n}\sum_{i = 1}^n\log p_\theta(x_i|x_1,\ldots,x_{i-1})
\xrightarrow[n\to\infty]{}  \E_{P_o}\log p_\theta(X_{k +1}|X_1,\ldots,X_k)$$
with probability 1, by the ergodic theorem \citep[][6.28]{Breiman_1968}.
Combining this with the small-sample correction (applied heuristically in this setting) suggests the approximation
\begin{align*}
\Pi\big(d\theta\mid d_n(X_{1:n},x_{1:n})<R\big) &\approxprop 
\exp\Big(-n\zeta_n\Big[-\Hcal(P_o)-\frac{1}{n}\sum_{i=1}^n \log p_\theta(x_i|x_1,\ldots,x_{i-1})\Big]\Big) \Pi(d\theta) \\
&\propto \Pi(d\theta)\prod_{i = 1}^n p_\theta(x_i|x_1,\ldots,x_{i-1})^{\zeta_n}.
\end{align*}
Thus, as in the i.i.d.\ case, the end result is an approximation obtained by simply 
raising the likelihood to the power $\zeta_n$.
In Section~\ref{section:autoregressive}, we apply this to perform robust inference for the order of an autoregressive model.



\subsection{Regression c-posterior based on conditional relative entropy}
\label{section:regression}

In regression, one observes covariates/predictors $x_1,\ldots,x_n$ associated with target values $y_1,\ldots,y_n$,
and models the conditional distribution of $y$ given $x$. As in the i.i.d.\ setting, in order to allow for contamination/misspecification,
let us suppose that $Y_i|x_i$ is drawn from the model $p_\theta(y|x)$ for $i = 1,\ldots,n$,
and the observed values $y_{1:n}$ are a slightly corrupted version of $Y_{1:n}$,
in the sense that $d_n(Y_{1:n},y_{1:n}|x_{1:n})<R$ for some measure of discrepancy $d_n(\cdot,\cdot|\cdot)$.
Suppose $(x_1,y_1),\ldots,(x_n,y_n)$ behave like i.i.d.\ samples from some $p_o(x,y)$.
For notational clarity, let us assume that these densities on $x$ and $y$ are with respect to measures
that we will denote by $d x$ and $d y$, respectively. 

A natural choice of discrepancy between the conditional distributions $p_o(y|x)$ and $p_\theta(y|x)$ is the conditional relative entropy,
$$D_\theta := \int p_o(x,y)\log\frac{p_o(y|x)}{p_\theta(y|x)}\,d x\,d y,$$
and in turn, an a.s.-consistent estimator of this quantity is a sensible choice for $d_n(\cdot,\cdot|\cdot)$.
Then, by Lemma \ref{lemma:limit}, the resulting c-posterior converges to a nice asymptotic form:
\begin{align*}
\Pi\big(d\theta\mid d_n(Y_{1:n},y_{1:n}|x_{1:n})<R\big)
&\Longrightarrow \Pi\left(d\theta\mid D_\theta < R\right) \propto \exp(-\alpha D_\theta) \Pi(d\theta) \\
&\propto \exp\Big(\alpha \int p_o(x,y)\log p_\theta(y|x)\,d x\,d y\Big) \Pi(d\theta)
\end{align*}
if we take $R\sim\Exponential(\alpha)$ as usual. To obtain an approximation that is applicable for smaller $n$ as well,
we apply the same small-sample correction as before, replacing $\alpha$ by $n\zeta_n$.
Combining this with an empirical approximation to the integral suggests using
\begin{align*}
\Pi(d\theta\mid d_n(Y_{1:n},y_{1:n}|x_{1:n})<R)
&\approxprop \exp\big(\zeta_n \textstyle{\sum_i} \log p_\theta(y_i|x_i)\big) \Pi(d\theta) \\
&= \Pi(d\theta) \prod_{i = 1}^n p_\theta(y_i|x_i)^{\zeta_n}.
\end{align*}
Consequently, once again, we arrive at a power posterior approximation to the c-posterior, 
allowing us to bypass the computation of $d_n(\cdot,\cdot|\cdot)$.
In Section~\ref{section:variable-selection}, we apply this to perform robust variable selection in linear regression.

\section{Applications}
\label{section:applications}



\subsection{Autoregressive models of unknown order}
\label{section:autoregressive}

We illustrate using the c-posterior to perform inference for the order of an autoregressive model in a way that is robust,
not only to the form of the
distribution of noise/shocks, but also to misspecification of the structure of the model, such as time-varying noise.
This serves as a nice demonstration of how the robustified marginal likelihood can be computed in closed form when using 
conjugate priors, and provides some insight into why coarsening works.

Consider an $\mathrm{AR}(k)$ model, that is, a $k$th-order autoregressive model:
$$ X_t = \sum_{\ell=1}^k \theta_\ell X_{t-\ell} + \epsilon_t $$
for $t = 1,\ldots,n$, where $\epsilon_1,\ldots,\epsilon_n\iid\N(0,\sigma^2)$ and $X_t = 0$ for $t\leq 0$ by convention.
Let $\pi(k)$ be a prior on the order $k$, let $\theta_1,\ldots,\theta_k|k \iid\N(0,\sigma_0^2)$,
and for simplicity, assume $\sigma^2$ is known.

We apply the time-series c-posterior 
developed in Section~\ref{section:time-series}
to obtain robustness to perturbations that are small in the sense of relative entropy rate.
Since $\beta|k$ has been given a conjugate prior, then as described in Section~\ref{section:conjugate-priors}, 
we can analytically compute the resulting marginal power likelihood,
\begin{align*}
p_c(x_{1:n}|k) & := \int_{\R^k} p(x_{1:n}|\theta,k)^{\zeta_n} \pi(\theta|k) d\theta \\
& = \int_{\R^k} \Big(\prod_{t=1}^n \N\big(x_t \,\big\vert\, {\textstyle\sum_{\ell=1}^k} \theta_\ell x_{t-\ell},\,\sigma^2\big)\Big)^{\zeta_n}
    \N(\theta \mid 0,\sigma_0^2 I_{k\times k}) d\theta \\
& = \frac{\exp(\tfrac{1}{2}\zeta_n^2 v^\T \Lambda^{-1} v)}{\sigma_0^k |\Lambda|^{1/2}} \N(x_{1:n}\mid0,\sigma^2 I_{n\times n})^{\zeta_n}
\end{align*}
where $\Lambda = \zeta_n M + \sigma_0^{-2} I_{k\times k}$, 
$M_{i j} = \sum_{t = 1}^n x_{t-i} x_{t-j} / \sigma^2$, and $v_i = \sum_{t = 1}^n x_t x_{t-i} / \sigma^2$, by straightforward calculation.
This, in turn, can be used to compute a robustified posterior on the model order $k$, defined as
$\pi_c(k|x_{1:n}) \propto p_c(x_{1:n}|k) \pi(k)$.
This is expected to be robust to departures from the $\mathrm{AR}(k)$ model that require more than $\alpha$ samples to distinguish,
and thus, it will favor values of $k$ that are consistent with the data to within this specified tolerance.

\begin{figure}
  \centering
  \includegraphics[trim=.5cm 1cm 1.6cm 0, clip, width=0.49\textwidth]{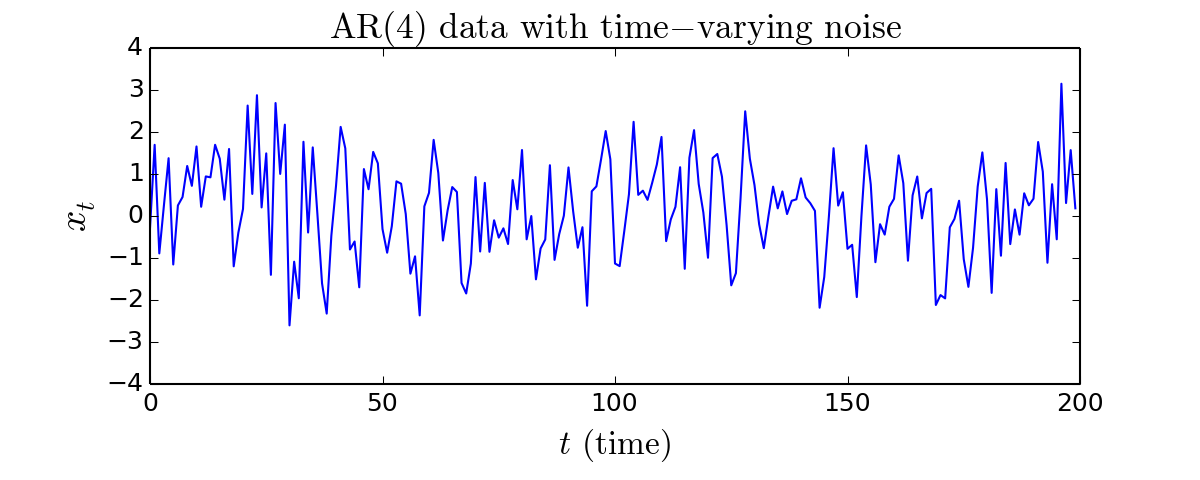}
  \includegraphics[trim=.5cm 1cm 1.6cm 0, clip, width=0.49\textwidth]{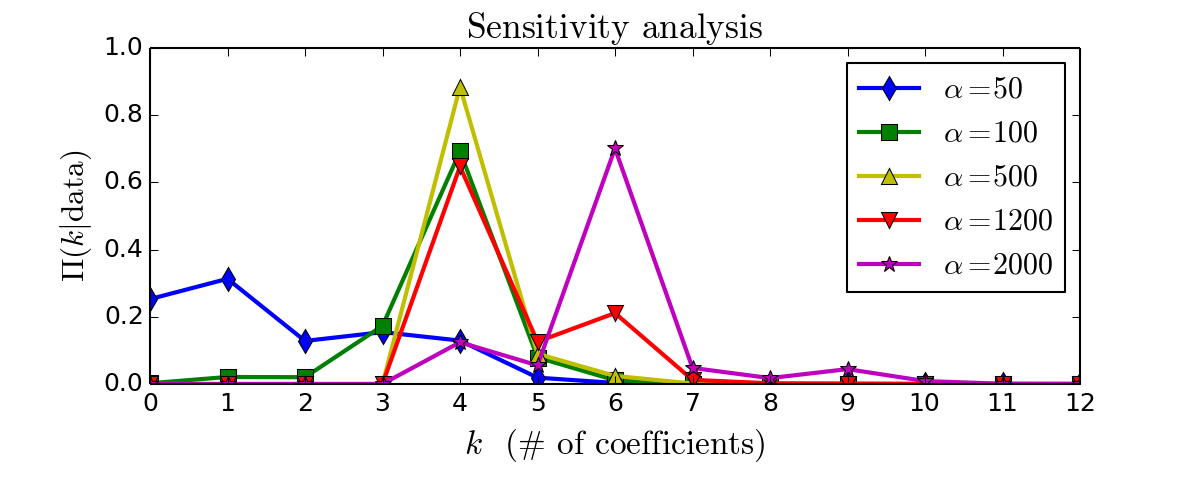}
  \includegraphics[trim=.5cm 1cm 1.75cm 0, clip, width=0.49\textwidth]{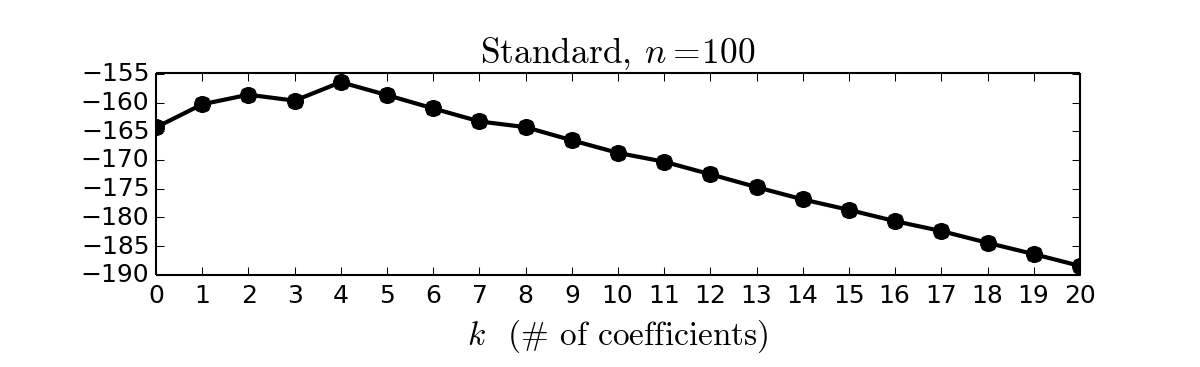}
  \includegraphics[trim=.5cm 1cm 1.75cm 0, clip, width=0.49\textwidth]{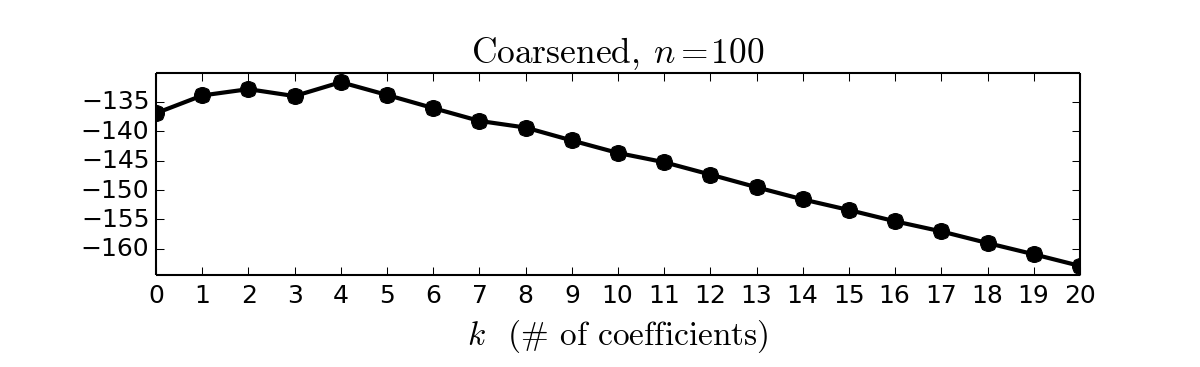}
  \includegraphics[trim=.5cm 1cm 1.75cm 0, clip, width=0.49\textwidth]{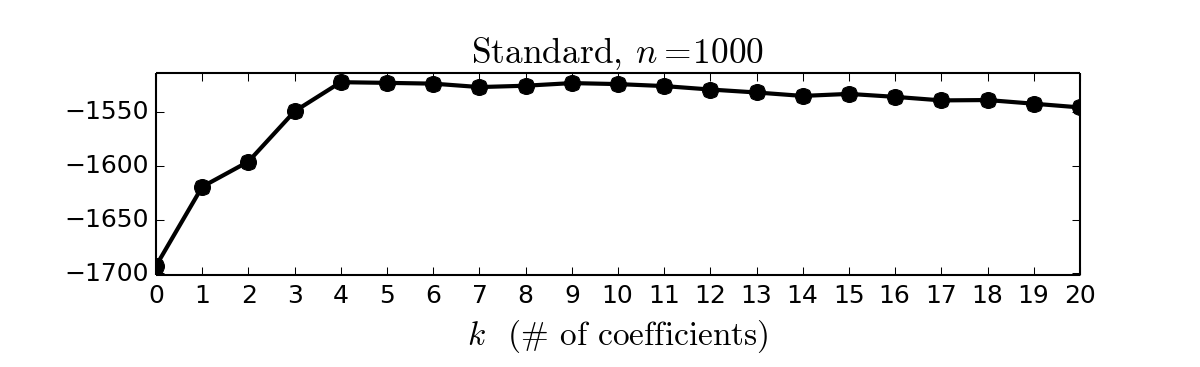}
  \includegraphics[trim=.5cm 1cm 1.75cm 0, clip, width=0.49\textwidth]{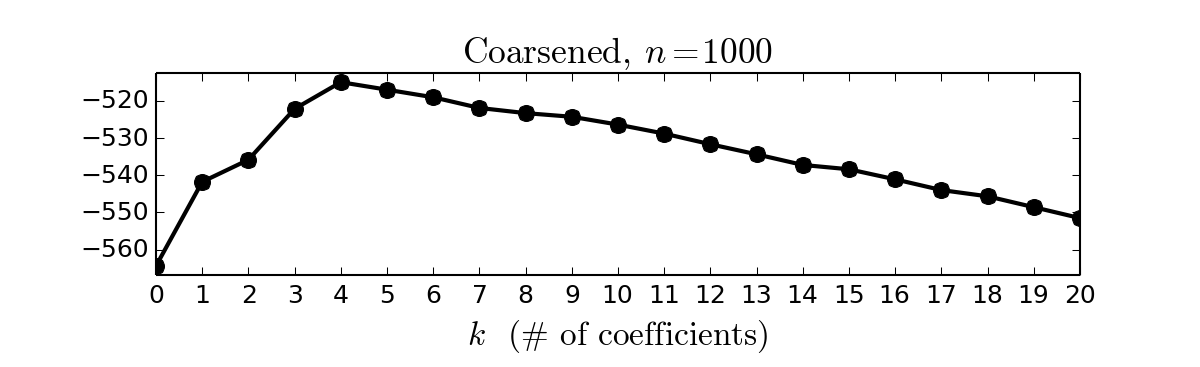}
  \includegraphics[trim=.5cm 0cm 1.75cm 0, clip, width=0.49\textwidth]{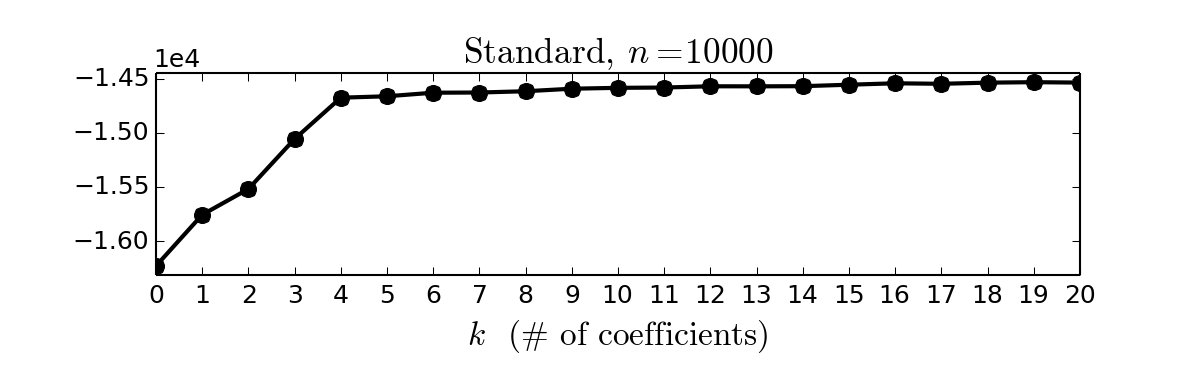}
  \includegraphics[trim=.5cm 0cm 1.75cm 0, clip, width=0.49\textwidth]{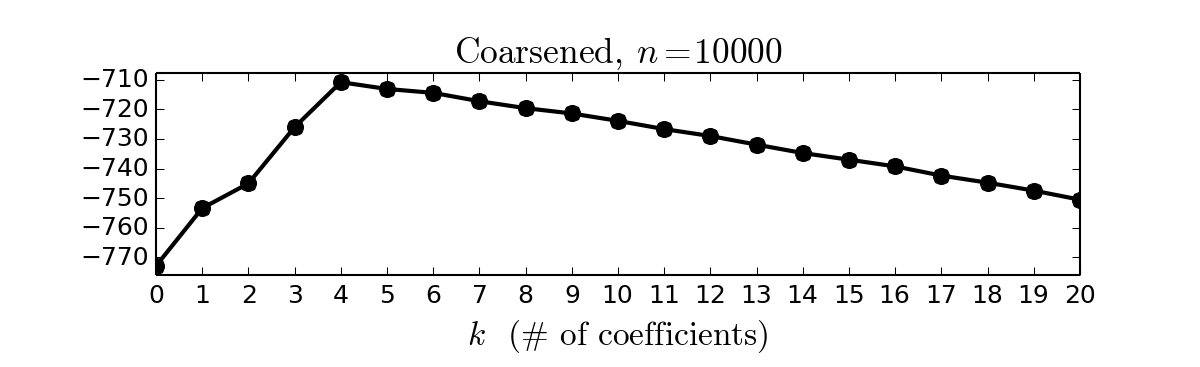}
  \caption{Autoregression example. 
  Upper left: Data sampled from the process in Equation \ref{equation:time-misspec}.
  Upper right: Sensitivity analysis, displaying the c-posterior on $k$ as $\alpha$ varies, when $n = 10^4$.
  Lower left: Log marginal likelihood of $\mathrm{AR}(k)$ model for $k = 0,1,\ldots,20$, on increasing amounts of data from this process.
  Lower right: Log of coarsened marginal likelihood for the same model, on the same data.
  }
  \label{figure:AR}
\end{figure}

To demonstrate empirically, we generate data from a process that is close to $\mathrm{AR}(4)$ but exhibits time-varying noise
that cannot be captured by the model:  
\begin{align}\label{equation:time-misspec}
x_t = \sum_{\ell=1}^4 \theta_\ell x_{t-\ell} + \epsilon_t + \tfrac{1}{2} \sin t
\end{align}
where $\theta = (1/4,1/4,-1/4,1/4)$, $\epsilon_t\iid\N(0,1)$, and $x_t = 0$ for $t\leq 0$.  
We apply the model above to such data, and compare the standard Bayesian approach to the coarsened approach.

For the model parameters, we set $\sigma^2 = 1$ to match the true value, and take $\sigma_0^2 = 1$.
If one expects a particular amount of misspecification, this can be used to make a principled choice of $\alpha$;
see Section~\ref{section:variable-selection}. Here, we assess sensitivity to the choice of $\alpha$, by considering the
c-posterior on $k$ as $\alpha$ varies, when $n = 10^4$, with an improper uniform prior on $k$; see Figure~\ref{figure:AR} (upper right).
There is a fairly wide range of values that give similar results---for
$\alpha$'s between 100 and 1200, the large majority of the mass is on the correct value of $k$, namely $k = 4$.

To visualize what happens as $n$ increases, we set $\alpha = 500$, and consider the log of the marginal likelihood.
Due to the misspecification, the standard posterior strongly favors values of $k$ much greater than 4 when $n$ gets sufficiently large;
see Figure~\ref{figure:AR} (lower left). 
Meanwhile, the c-posterior stabilizes to a distribution on $k$ favoring $k = 4$;
see Figure~\ref{figure:AR} (lower right).
For values of $n$ less than $\alpha$, the standard and coarsened approaches yield similar results, 
however, as $n$ increases beyond $\alpha$, they differ markedly.

More generally, this type of picture is typical for the log marginal likelihood when comparing models of increasing complexity.
As discussed in Section~\ref{section:robustness}, the log marginal likelihood automatically penalizes more complex models, via the term
$-\tfrac{1}{2}t_k \log n$ where $t_k$ is the dimension of the parameter space, e.g., $t_k = k$ for the $\mathrm{AR}(k)$ model above;
this penalty is visible in the linear decline exhibited in the $n=100$ plot. 
As $n$ increases, this complexity penalty only increases proportionally to $\log n$,
and thus it becomes overwhelmed by the main term of order $n$,
involving the log-likelihood at the maximum likelihood estimator within model $k$.
When $n$ is sufficiently large, the following pattern emerges, as seen in the $n=10000$ plot (for the standard approach):
for model complexity values $k$ that are too small, there is a clear lack of fit, and 
as $k$ increases the log marginal likelihood increases rapidly until the model can fairly closely approximate the data distribution,
at which point it plateaus, continuing to increase only slightly after that as only fine grain improvements can be made.

From this perspective, the reason why the coarsened marginal likelihood ``works'' is that when $n$ is large,
it maintains a balance between the model complexity penalty and the main log-likelihood term,
by behaving as though the sample size is no larger than $\alpha$.


\subsection{Variable selection in linear regression}
\label{section:variable-selection}

Consider the following spike-and-slab model for variable selection:
\begin{align*}
    & W\sim\Beta(r,s) \\
    & \beta_j\sim\N(0,1/L_0) \text{ with probability } W, \text{ otherwise } \beta_j = 0, \text{ for each } j=1,\ldots,p \\
    & \lambda \sim \Ga(a,b) \\
    & Y_i|\beta,\lambda\,\sim\,\N(\beta^\T x_i,1/\lambda) \text{ independently for } i=1,\ldots,n.
\end{align*}
Models of this type are often used to infer which covariates $x_{i 1},\ldots,x_{i p}$ are predictive of the target variable $y_i$,
by considering which coefficients $\beta_j$ have a high posterior probability of being nonzero.
This provides valuable insight into the relationships present in the data generating process.
However, usually, it is unlikely that the data exactly follow the $\N(\beta^\T x_i,1/\lambda)$ form,
and although the model exhibits some robustness to departures from normality, it is not robust to departures from 
the linearity assumed in the mean function $\beta^\T x_i = \beta_1 x_{i 1} + \cdots + \beta_p x_{i p}$. 
For instance, if the mean is actually $\beta_1 g(x_{i 1})$ where $g$ is close to but not exactly linear, 
and $x_{i 1}$ is correlated with other covariates,
then the posterior will typically make additional coefficients nonzero in order to compensate.

We demonstrate how the c-posterior provides robustness to misspecification of this type. 
This example also provides an opportunity to show how Gibbs sampling can be used with power posteriors
when conditionally-conjugate priors have been chosen.

As described in Section~\ref{section:regression}, in the regression setting, 
the c-posterior based on conditional relative entropy can be approximated by 
the power posterior obtained by raising the likelihood to $\zeta_n = 1/(1 + n/\alpha)$, as before.
If we first integrate $W$ out of the model, the resulting power posterior is
$$ \pi_c(\beta,\lambda|y) \propto \pi(\beta,\lambda) p(y|\beta,\lambda)^{\zeta_n}. $$
Due to the use of conditionally-conjugate priors, the full conditionals for $\beta_j$ and $\lambda$ can be derived in closed form,
by standard calculations. We give the formulas here without justification:
$$ \pi_c(\lambda|\beta,y) = \Ga\Big(\lambda \,\Big\vert\, a + \tfrac{1}{2} n \zeta_n,\,
b + \tfrac{1}{2}\zeta_n \textstyle\sum_{i=1}^n (y_i - \beta^\T x_i)^2 \Big) $$
and one can sample from $\pi_c(\beta_j|\beta_{-j},\lambda,y)$, where $\beta_{-j} = (\beta_\ell : \ell\neq j)$, by setting $\beta_j = 0$ with 
probability 
$$ \Pi_c(\beta_j = 0 \mid \beta_{-j},\lambda,y) = \left(1 + \sqrt{L_0/L} \,\exp\big(\tfrac{1}{2} L M^2\big)
\frac{r+\sum_{\ell\neq j} \I(\beta_\ell \neq 0)}{s+\sum_{\ell\neq j}\I(\beta_\ell = 0)} \right)^{-1}$$
where $L = L_0 + \lambda \zeta_n\sum_{i=1}^n x_{i j}^2$, $M = (\lambda \zeta_n / L) \sum_{i=1}^n \delta_i x_{i j}$,
and $\delta_i = y_i - \sum_{\ell\neq j} \beta_\ell x_{i \ell}$, and otherwise sampling $\beta_j$ from $\N(M,L^{-1})$.

\subsubsection{Simulation example}

To demonstrate empirically, first consider a simulated example where the mean of the observed data
is a slightly nonlinear function of a single covariate, plus a constant offset:
\begin{align}\label{equation:varsel-data}
    y_i = \beta_{0 1} + \beta_{0 2} (x_{i 2} + \tfrac{1}{16}x_{i 2}^2) + \epsilon_i
\end{align}
where $\beta_{0 1} = -1$, $\beta_{0 2} = 4$, and $\epsilon_1,\ldots,\epsilon_n\iid\N(0,1)$.
Following standard practice, suppose $x_{i 1} = 1$, to accomodate a constant offset.
Suppose there are five covariates $x_{i 2},\ldots,x_{i 6}$ distributed according to a multivariate skew-normal distribution
\citep{Azzalini_1999} which has been centered and scaled so that each covariate has zero mean and unit variance:
$X_{i j} = (\tilde X_{i j} - \E \tilde X_{i j} ) / \sigma(\tilde X_{i j})$ for $j = 2,\ldots,6$, where
$\tilde X_i \sim \SN_5(\Omega,a)$ with shape $a = (0.6,2.7,-3.3,-4.9,-2.5)$ and scale matrix  
$$\Omega = \begin{pmatrix}
1.0 & -0.89 & 0.93 & -0.91 & 0.98\\
-0.89 & 1.0 & -0.94 & 0.97 & -0.91\\
0.93 & -0.94 & 1.0 & -0.96 & 0.97\\
-0.91 & 0.97 & -0.96 & 1.0 & -0.93\\
0.98 & -0.91 & 0.97 & -0.93 & 1.0 \end{pmatrix}.
$$
The $a$ and $\Omega$ above were randomly-generated; there is nothing particularly special about them,
except that $\Omega$ was chosen so that the covariates would be fairly strongly correlated.
Figure~\ref{figure:varsel} (top) shows a scatterplot of $y_i$ versus $x_{i 2}$ for 200 samples, as well as the mean as a function of $x_{i 2}$.

\begin{figure}
  \centering
  \includegraphics[trim=.5cm 0 1.75cm 0, clip, width=0.49\textwidth]{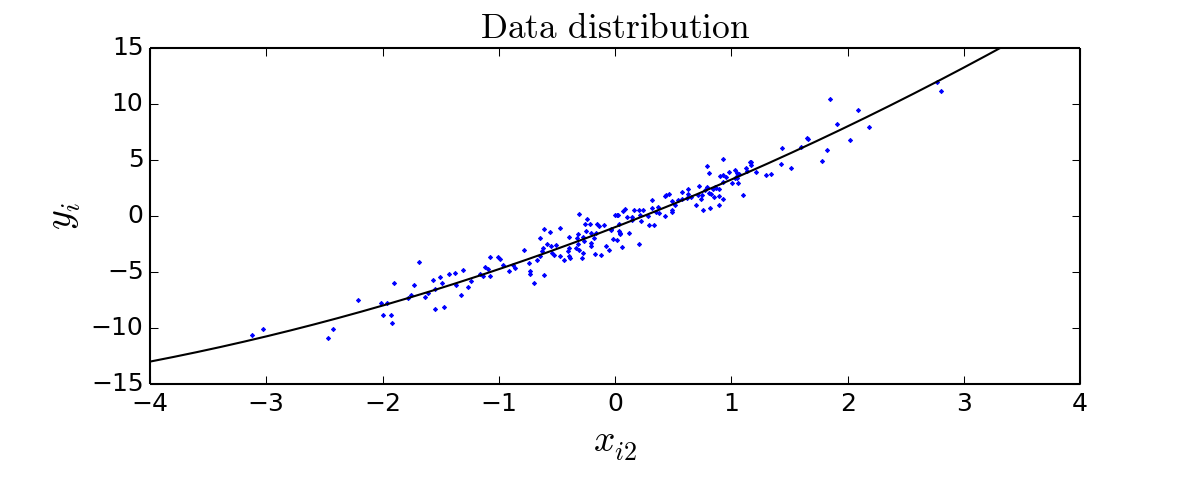} \\
  \includegraphics[trim=.5cm 0 4.9cm 0, clip, height=0.23\textwidth]{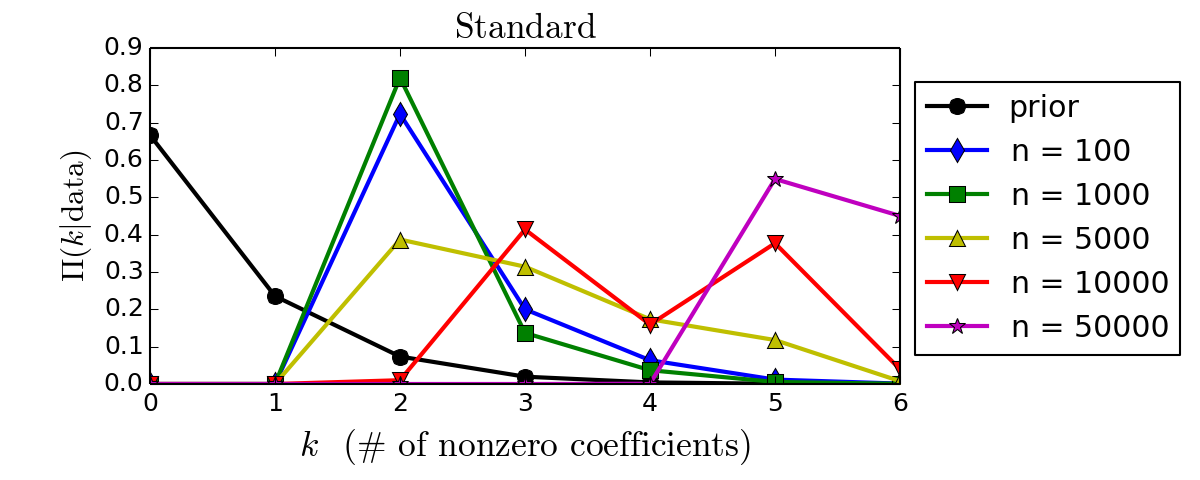}
  \includegraphics[trim=.5cm 0 .2cm 0, clip, height=0.23\textwidth]{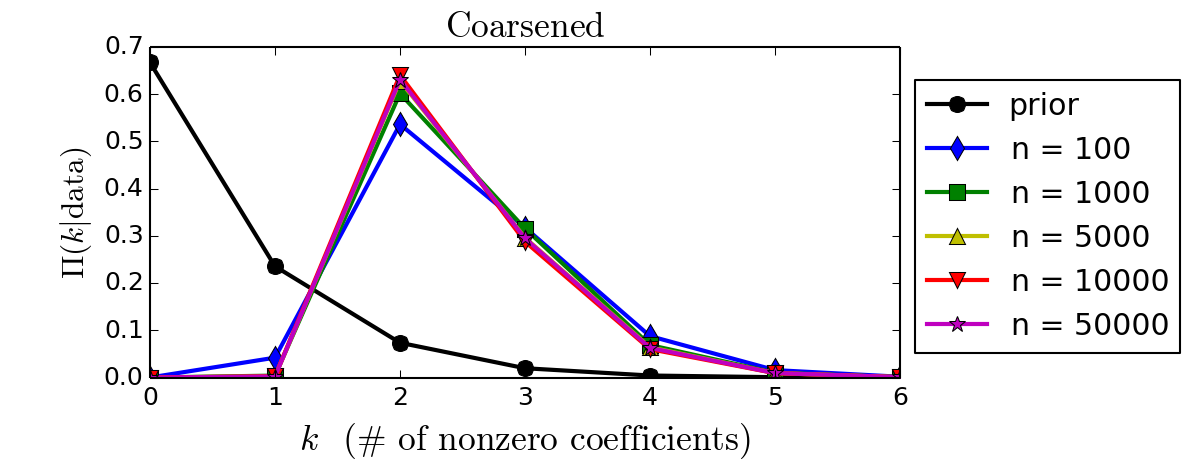}
  \includegraphics[trim=1.1cm 1cm 1.7cm 0, clip, width=0.49\textwidth]{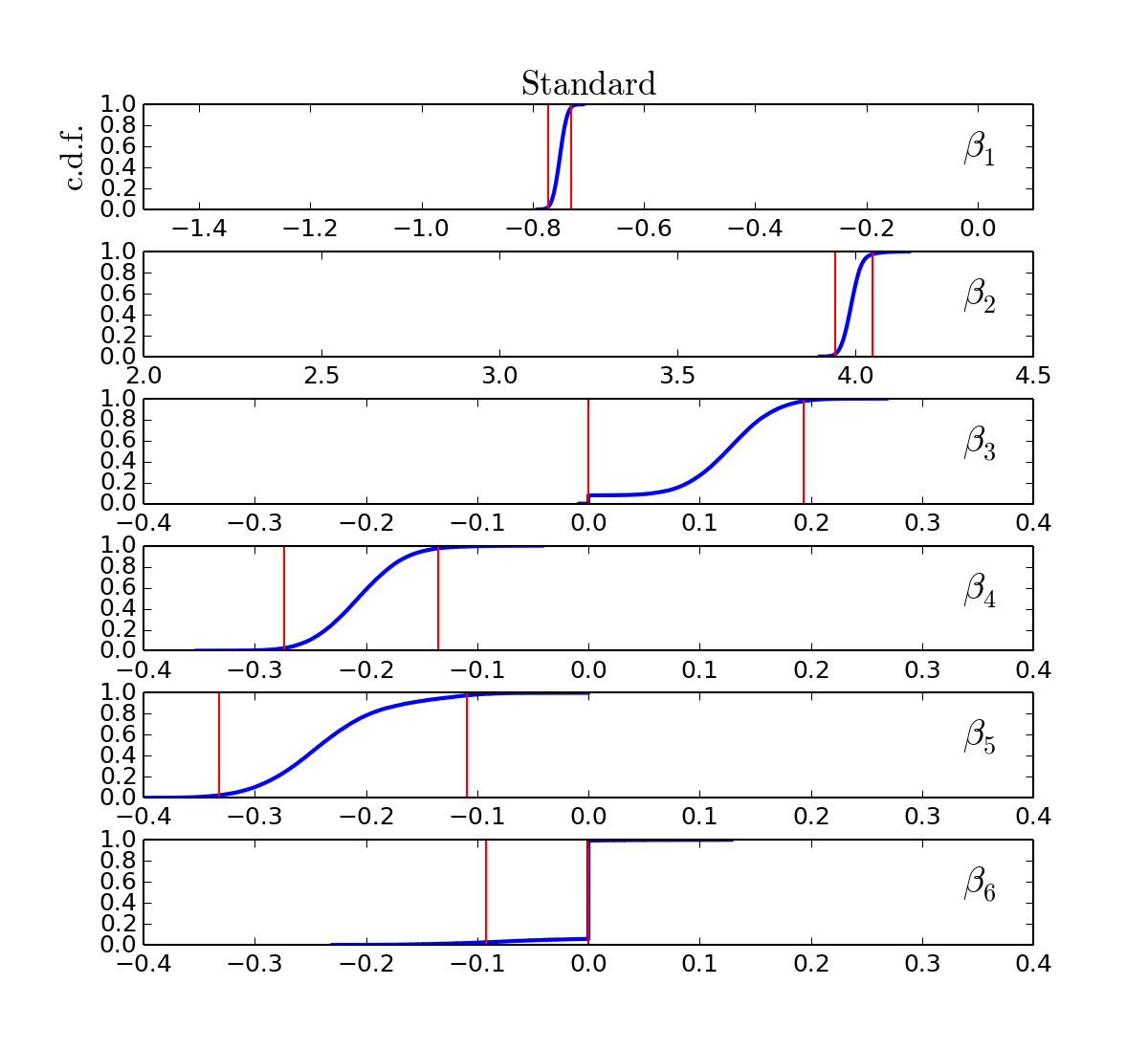}
  \includegraphics[trim=1.1cm 1cm 1.7cm 0, clip, width=0.49\textwidth]{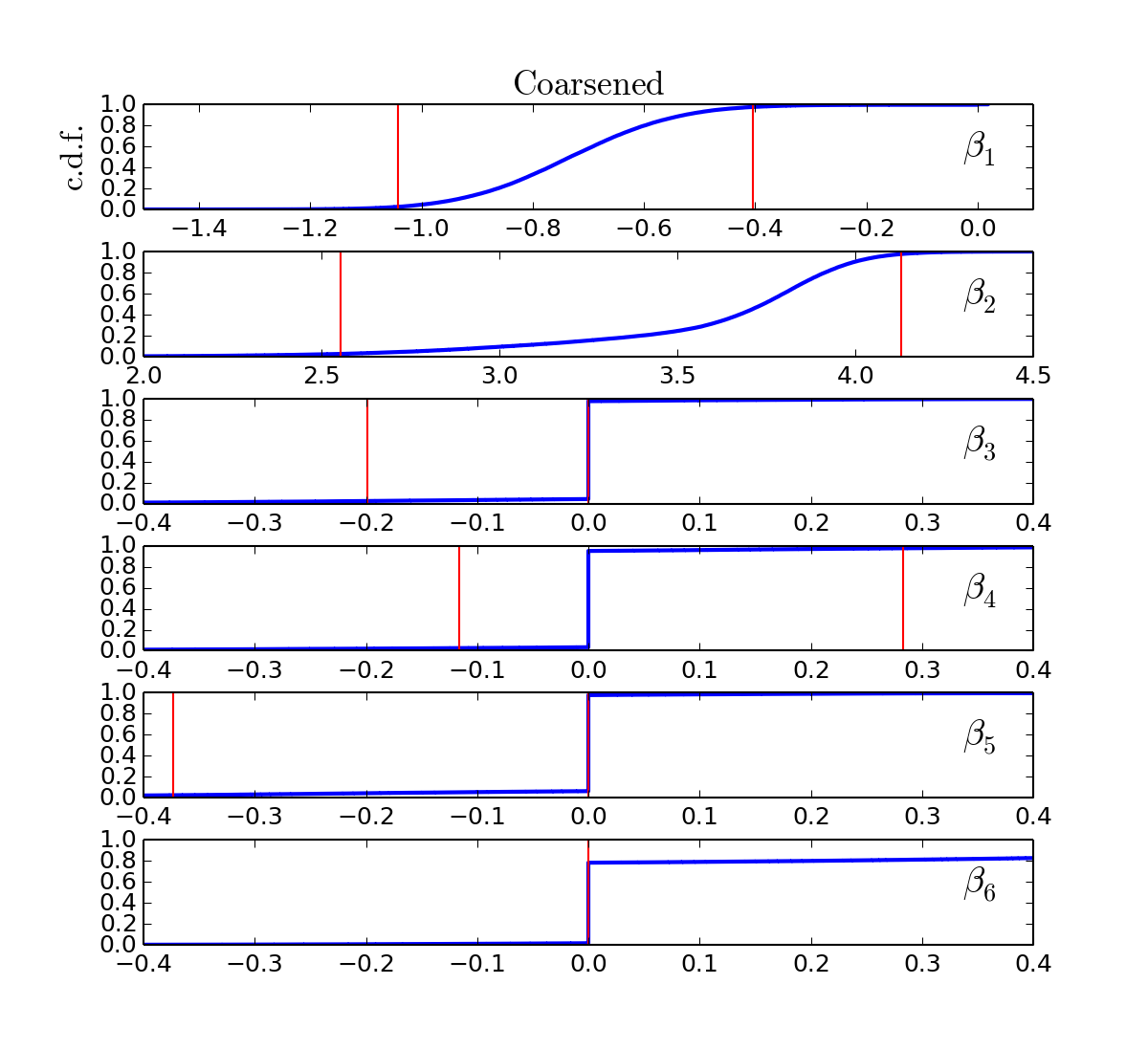}
  \caption{Variable selection with simulated data. 
  Top: Scatterplot of the target variable $y_i$ versus $x_{i 2}$, as well as the mean function (black line).
  Middle left: The posterior on the number of nonzero coefficients $k$ favors larger values as $n$ increases.
  Middle right: The c-posterior favors the ``true'' number, $k = 2$, even as $n$ grows.
  Bottom left: Posterior c.d.f.\ for each coefficient (blue), and 95\% credible interval (red).
  Bottom right: Same, for the c-posterior.
  }
  \label{figure:varsel}
\end{figure}

For the model parameters, we choose $r = 1$ and $s = 2 p$ (in order to favor having $O(1)$ nonzero coefficients, regardless of $p$), 
$L_0 = 1$, and $a = b = 1$.
To choose the coarsening parameter $\alpha$, recall that the c-posterior is obtained by conditioning on 
the (estimated) conditional relative entropy being less than $R$, where $R\sim\Exp(\alpha)$. 
The relative entropy between two Gaussians $\N(\mu_1,\sigma^2)$ and $\N(\mu_2,\sigma^2)$ is $\frac{1}{2\sigma^2}(\mu_1 - \mu_2)^2$.
Thus, if we expect the misspecification/contamination to shift the mean function by approximately $\pm\delta$ on average, 
and the noise has standard deviation $\sigma$, then it is reasonable to choose $\alpha$ so that $\E R \approx \delta^2 / (2\sigma^2)$,
i.e., $\alpha \approx 2\sigma^2 / \delta^2$.  
In the present situation, by cheating and using our knowledge of the truth,
we choose $\delta = 0.2$ and $\sigma = 1$, leading to $\alpha = 50$.

For each $n\in\{100,1000,5000,10000,50000\}$, ten datasets were generated, and for both the standard posterior and the coarsened posterior,
$50000$ Gibbs sweeps were performed on each dataset, the first $5000$ of which were discarded as burn-in.  
  
Figure~\ref{figure:varsel} (middle) 
shows the average of these posteriors on $k$ over the 10 datasets, for the standard and coarsened posteriors.
Note that the ``true'' number of nonzero coefficients in Equation \ref{equation:varsel-data} is $k = 2$ ($\beta_{0 1}$ and $\beta_{0 2}$).

Figure~\ref{figure:varsel} (bottom) shows the posterior cumulative distribution function (c.d.f.) and 95\% credible interval
for each coefficient $\beta_1,\ldots,\beta_6$ when $n = 10000$, for the standard and coarsened posteriors.
Recall that the ``true'' values are $\beta_1 = -1$, $\beta_2 = 4$, and $\beta_3 =\cdots =\beta_6 = 0$.
The 95\% intervals for the standard posterior are quite far from the true values of $\beta_1$, $\beta_4$, and $\beta_5$,
while all of the 95\% intervals for the c-posterior contain the true values;
also note that for $\beta_3,\ldots,\beta_6$, most of the c-posterior probability is at zero.
The case of $\beta_1$, in particular, illustrates that in addition to incorrectly inferring which coefficients are nonzero,
under misspecification, the standard posterior can also lead to incorrect inferences about the values of the nonzero coefficients.
The c-posterior mitigates this by more appropriately calibrating the amount of concentration;
however, the price to be paid is that this can cause the c-posterior to be considerably more diffuse than necessary,
such as in the case of $\beta_2$.

\subsubsection{Modeling birthweight of infants}

The Collaborative Perinatal Project (CPP) collected data from a large sample of mothers and their children, measuring many 
medical and socioeconomic variables from before and during pregnancy, as well as early childhood \citep{klebanoff2009collaborative}.
Using a subset of the CPP data, we illustrate how the c-posterior can be used to analyze the relationship between 
birthweight and a number of predictor variables.

The dataset we use contains $n = 2379$ subjects, and $71$ covariates that are potentially predictive of birthweight.
The data are preprocessed to normalize each covariate as well as the target variable, 
by subtracting off the sample mean and dividing by the sample standard deviation for each.
As usual, a constant covariate is appended, making $p = 72$.
We use the same model parameters as in the simulation example.
Rather than choose a single value of $\alpha$, we explore the data at varying levels of coarseness, by considering a range of $\alpha$ values.

\begin{figure}
  \centering
  \includegraphics[trim=.5cm 0 1.75cm 0, clip, width=0.49\textwidth]{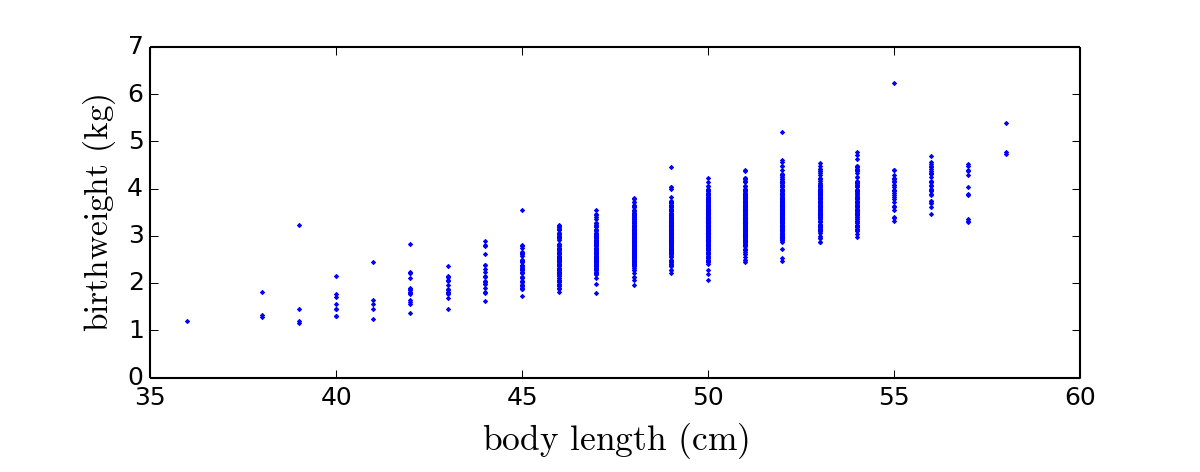}
  \includegraphics[trim=.5cm 0 1.75cm 0, clip, width=0.49\textwidth]{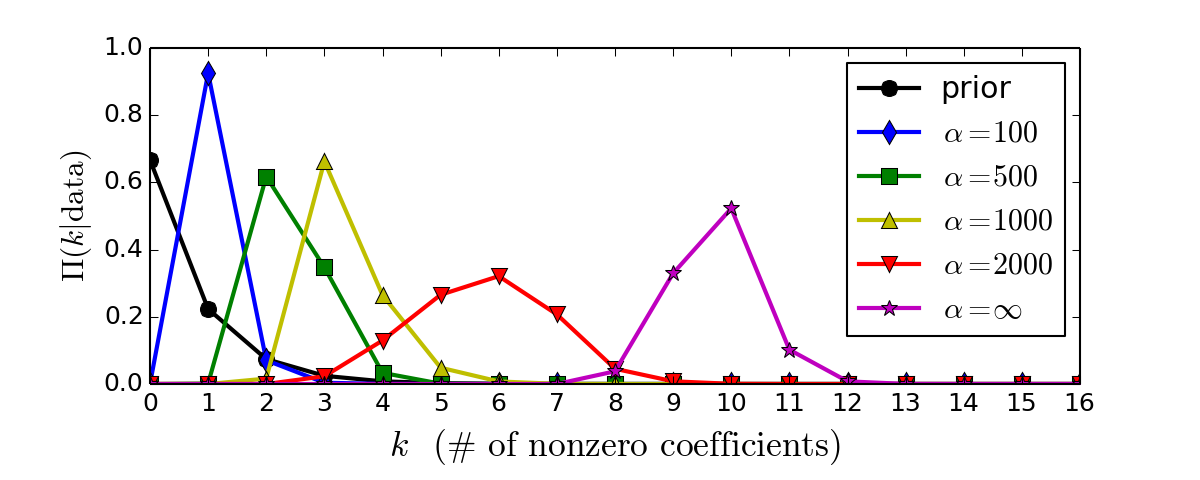}
  \includegraphics[trim=.5cm 0 1.75cm 0, clip, width=0.49\textwidth]{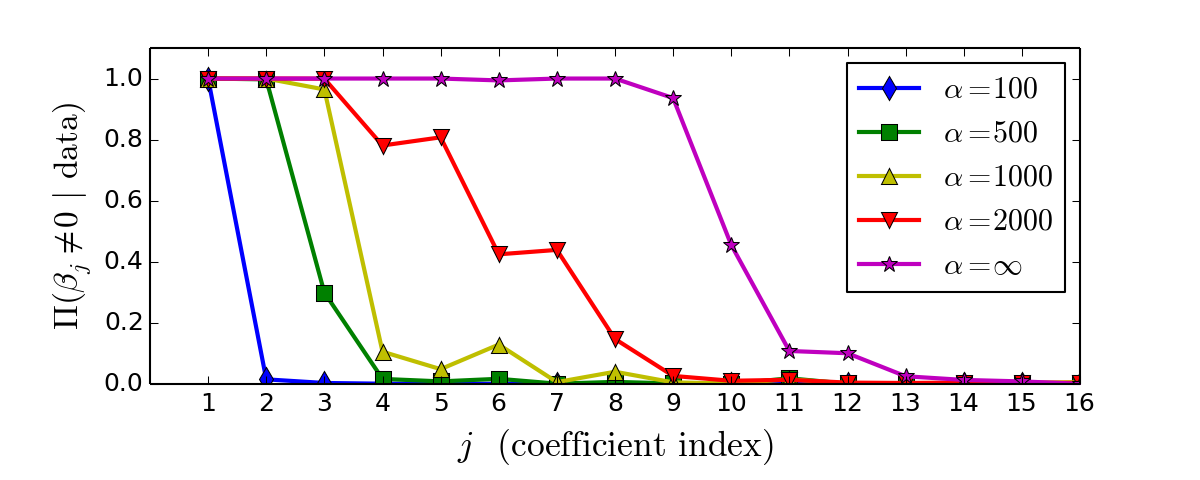}
  \includegraphics[trim=.5cm 0 1.75cm 0, clip, width=0.49\textwidth]{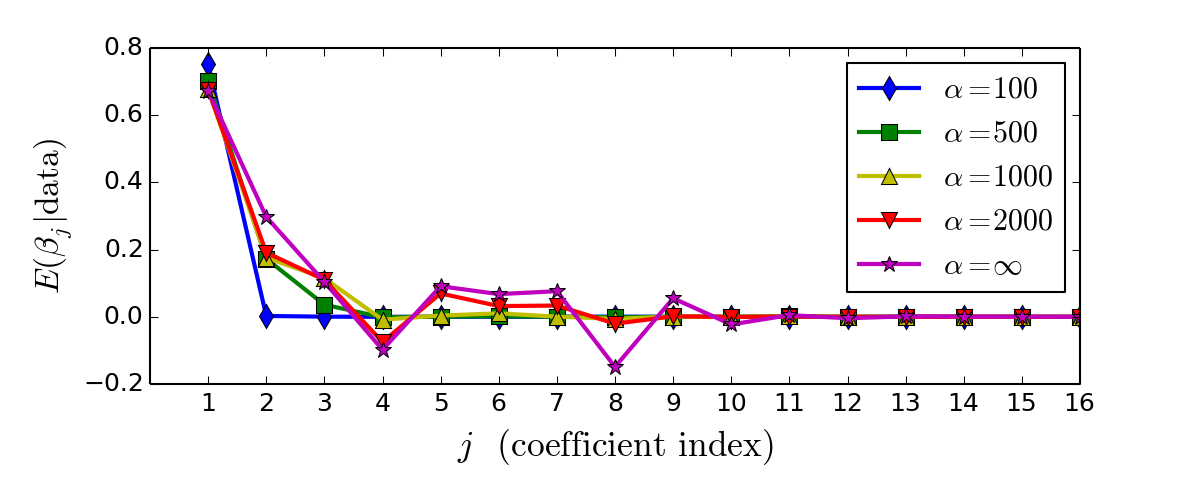}
  \caption{Variable selection for modeling birthweight. 
  Upper left: Scatterplot of birthweight (the target variable) versus body length at birth.
  Upper right: The standard posterior includes several more nonzero coefficients than the c-posteriors.
  Lower left: Posterior probability of inclusion for each coefficient; only the top 16 are shown, see list below.
  Lower right: Posterior mean of each coefficient, for the same 16.
    (Top 16 variables:
    1.\ Body length,
    2.\ Mother's weight at delivery,
    3.\ Gestation time,
    4.\ African-American,
    5.\ Center 6,
    6.\ Center 2,
    7.\ Center 3,
    8.\ Mother's weight prepregnancy,
    9.\ Previously pregnant,
    10.\ Cigarettes per day,
    11.\ \# prenatal checkups,
    12.\ Smoker/non-smoker,
    13.\ Mother's BMI prepregnancy,
    14.\ \# previous pregnancies,
    15.\ Triglyceride level,
    16.\ Center 10.)
  }
  \label{figure:cpp}
\end{figure}

For each $\alpha\in\{100, 500, 1000, 2000,\infty\}$, we run the sampler for $10^4$ Gibbs sweeps, discarding the first $1000$ sweeps as burn-in.
The sampler is initialized by setting all the coefficients to zero; initializing with a sample from the prior yields identical results.
To interpret $\alpha$ in terms of Euclidean notions, we estimate from posterior samples that $\lambda\approx 2.5$ to $3$, and thus,
by the formula $\alpha \approx 2\sigma^2 / \delta^2 = 2/(\lambda\delta^2)$ derived above,
the values of $\alpha$ above roughly correspond to allowing for misspecification/contamination of magnitude
$\delta\in\{0.09,0.04,0.03,0.02,0\}$, respectively,
or, when scaled to the original units, roughly $\delta_{\mathrm{kg}}\in\{0.045,0.02,0.015,0.01,0\}$ kilograms.

The posterior on the number of nonzero coefficients $k$ (Figure~\ref{figure:cpp}, upper right), 
and the posterior probability of inclusion (Figure~\ref{figure:cpp}, lower left), show that
the standard posterior includes around 10 out of the 72 coefficients, while the c-posterior employs a more parsimonious representation,
depending on $\alpha$.
At $\alpha = 100$ ($\delta_{\mathrm{kg}}\approx 0.045$), typically only a single variable is included, namely, body length.
It makes sense that body length would be strongly predictive of weight, and 
the scatterplot in Figure~\ref{figure:cpp} (upper left) confirms this.
At $\alpha = 500$ ($\delta_{\mathrm{kg}}\approx 0.02$), both body length and mother's weight at delivery are included, 
as well as gestation time, with somewhat lower probability; again, it makes sense for these to be predictive of birthweight.
As $\alpha$ increases, additional variables are included to account for finer aspects of the data,
until we reach the standard posterior at $\alpha = \infty$.

All of the variables included by the standard posterior could conceivably be predictive of birthweight,
although after adjusting for primary variables such as body length, it is possible that they are only being included due to misspecification.
Since it seems likely that there would be misspecification at least at the $\delta_{\mathrm{kg}} = 0.01$ kg level
(i.e., $\approx 0.02$ pounds, $\alpha = 2000$), the high probability placed on some of the additional variables
included by the standard posterior is dubious, 
as is the precision with which it purports to infer the corresponding coefficients.
Further, it seems inevitable that if $n$ were larger, then even more coefficients would be included by the standard posterior.
If there is misspecification, then as $n$ grows, eventually the interpretation of which coefficients are included, and their values,
becomes less related to practically significant associations and more related to the fact that the model is compensating for its limitations.

\subsection{Mixture models with a prior on the number of components}
\label{section:mixtures}

Consider a finite mixture model:
$$ X_1,\ldots,X_n|k,w,\varphi \,\iid\, \sum_{i=1}^k w_i f_{\varphi_i}(x), $$
and place a prior $\pi(k,w,\varphi)$ on the number of components $k$, the mixture weights $w$, and the component parameters $\varphi$.
This type of model is not robust to misspecification of the family of component distributions $(f_\varphi:\varphi\in\Phi)$,
resulting in negative consequences in practice,
since we might reasonably expect the observed data $x_1,\ldots,x_n$ to come from a finite mixture, 
but it is usually unreasonable to expect the component distributions to have a nice parametric form.
We illustrate how the c-posterior enables one to perform inference for the number of components, 
as well as the mixture weights and the component parameters,
in a way that is robust to misspecification of the form of the component distributions.
This example also serves to demonstrate the use of Metropolis--Hastings MCMC for inference with a power posterior.

We approximate the relative entropy c-posterior using the power posterior, defined as
$$ \pi_c(k,w,\varphi|x_{1:n}) \propto \pi(k,w,\varphi) \prod_{j = 1}^n\Big(\sum_{i=1}^k w_i f_{\varphi_i}(x_j)\Big)^{\zeta_n}. $$
The usual approaches to inference in mixtures are based on latent variables indicating which component each datapoint comes from,
but unfortunately, that does not seem to work here.  There are, nonetheless, a few possible approaches to doing inference,
for instance, \citet{walker13bayesian} developed an auxiliary variable technique for posteriors of this form,
and it would also be possible to use reversible jump MCMC \citep{green1995reversible}.
To keep things as simple as possible, however, we assume an upper bound on $k$, say $k\leq m$, and reparameterize the 
model in a way that enables one to simply use plain-vanilla Metropolis--Hastings (MH) MCMC on a fixed-dimensional space.
Specifically, we rewrite the mixture density as $\sum_{i=1}^m w_i f_{\varphi_i}(x)$
where $w_i = g(v_i)/\sum_{j = 1}^m g(v_j)$ and $g(v) =\max\{v - c,\,0\}$, so that $w_i = 0$ if $v_i \leq c$.
Letting $v_1,\ldots,v_m\sim\Ga(a,b)$ i.i.d., conditioned on the event that $\sum_{i = 1}^m g(v_i)>0$,
and letting $\varphi_1,\ldots,\varphi_m$ be i.i.d., yields a mixture model in which the prior on the number of components $k$
(that is, the number of nonzero weights)
is $\pi(k)\propto\Binomial(k|m,p)\I(k>0)$ where $p = \Pr(v_i > c)$.

\subsubsection{Skew-normal mixture example}

To demonstrate robustness to the form of the component distributions, we consider a univariate Gaussian mixture model, applied to data 
generated i.i.d.\ from the two-component mixture $\tfrac{1}{2}\SN(-4,1,5) +\tfrac{1}{2}\SN(-1,2,5)$, where $\SN(\xi,s,a)$ 
is the skew-normal distribution \citep{Azzalini_1999} with location $\xi$, scale $s$, and shape $a$; see Figure~\ref{figure:skewmix} (top).
For the model parameters, we assume an upper bound of $m = 10$ components, and define the prior on the component means and precisions as
$\mu_i\sim\N(0,5^2)$ and $\log(\lambda_i)\sim\N(0,2^2)$ independently for $i=1,\ldots,m$,
where the component densities are of the form $f_{\mu,\lambda}(x) =\N(x|\mu,\lambda^{-1})$.
We use a prior on $k$ and $w$ as in the preceding paragraph, with $a = 1/m$, $b = 1$, and $c$ such that $p = \Pr(v_i>c) = 1/m$.

For the c-posterior, we set $\zeta_n = 1/(1 + n/\alpha)$ (following Equation \ref{equation:zeta}), and choose $\alpha = 100$;
this can be interpreted as saying 
that we want the posterior to behave as though at most $100$ samples are available.
The top panel in Figure~\ref{figure:skewmix} illustrates that based on 100 samples, one can visually determine that there are two large groups
of roughly equal size, and can roughly determine their location and scale, but cannot determine their precise form---in particular,
one cannot tell that they are not actually Gaussian.

For each $n\in\{20,100,500,2000,10000\}$, five independent datasets of size $n$ were generated, and 
$10^5$ MH sweeps were performed for both the standard and coarsened posteriors.
Each sweep consisted of MH moves on each $(\mu_i,\lambda_i)$ and $v_i$ separately.
Figure~\ref{figure:skewmix} (left) shows the average of the approximated posteriors on $k$ over these five datasets.
For typical samples from the standard and coarsened posteriors when $n=10000$, Figure~\ref{figure:skewmix} (right) shows the
mixture density $\sum_{i=1}^m w_i f_{\mu_i,\lambda_i}(x)$, and the individual (weighted) components $w_i f_{\mu_i,\lambda_i}(x)$.

As expected, since the data distribution cannot be represented as a finite mixture of Gaussians, 
the standard posterior introduces more and more components as $n$ increases, in order to obtain an adequate fit to the data.
Meanwhile, in accordance with our visual intuition that, based on 100 samples, there appear to be two large groups, 
the $\alpha = 100$ c-posterior on $k$ shows strong support for two components, no matter how large $n$ becomes.

The typical sample from the standard posterior provides a better fit to the distribution of the data, however,
it has several more than two components, obscuring the fact that there are two large groups.
On the other hand, the typical sample from the c-posterior does not fit the data distribution as well,
but it clearly indicates that there are two groups, and provides an interpretable representation of their locations and scales.

\subsubsection{Shapley galaxy dataset}

The galaxy dataset of \citet{Roeder_1990} is a classic benchmark for nonparametric mixture models, but it is 
somewhat outdated, and rather small, with only $n=82$. 
\citet{Drinkwater_2004} provide a more recent and larger dataset of the same type, consisting of the velocities of 
$4215$ galaxies in the Shapley supercluster, a large concentration of gravitationally-interacting galaxies;
see Figure~\ref{figure:shapley}. The clustering tendency of galaxies continues to be a subject of interest in astronomy. 
However, due to the filament-like nature of the distribution of galaxies, it seems likely that any such clusters will not be Gaussian.

Nonetheless, with the c-posterior approach, a Gaussian mixture model can be used to good effect, to identify clusters
that are approximately normal. By varying the coarsening parameter $\alpha$, one can explore the data
at varying levels of precision, allowing for greater or smaller departures from normality. 

\begin{figure}
  \centering
  \includegraphics[trim=.5cm 0 1.75cm 0, clip, width=0.49\textwidth]{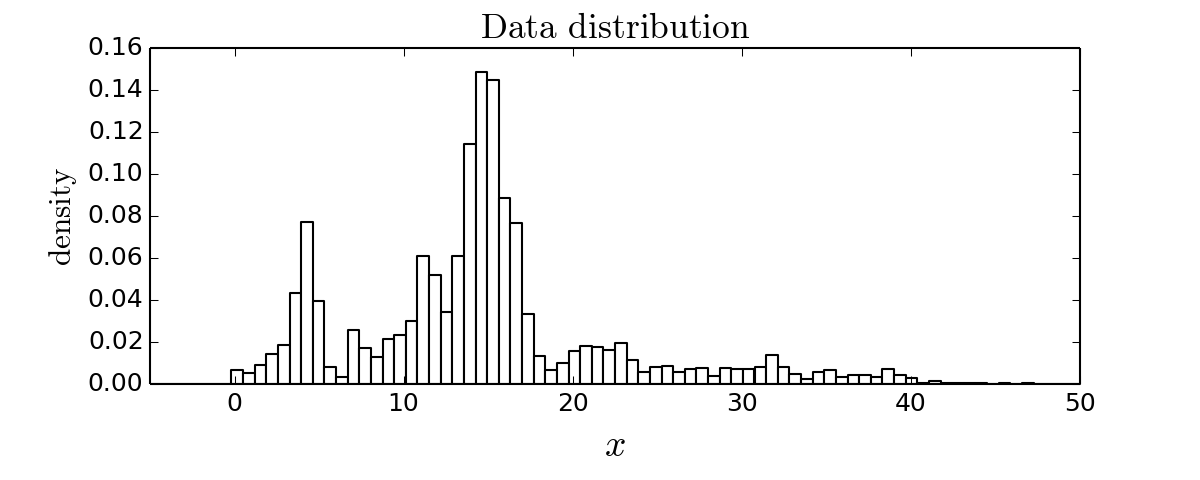}
  \includegraphics[trim=.5cm 0 1.75cm 0, clip, width=0.49\textwidth]{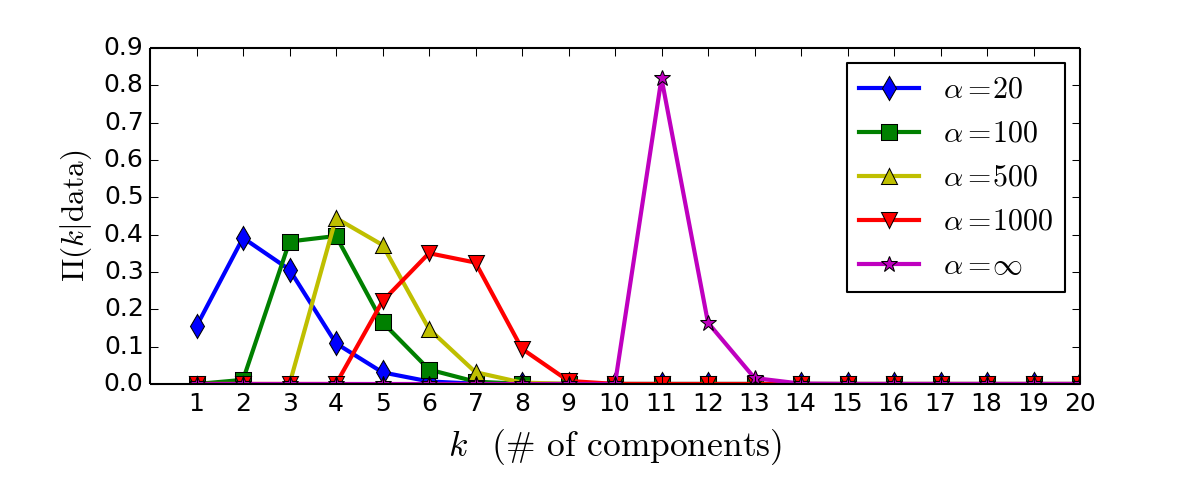}
  \includegraphics[trim=.5cm 0 1.75cm 0, clip, width=0.49\textwidth]{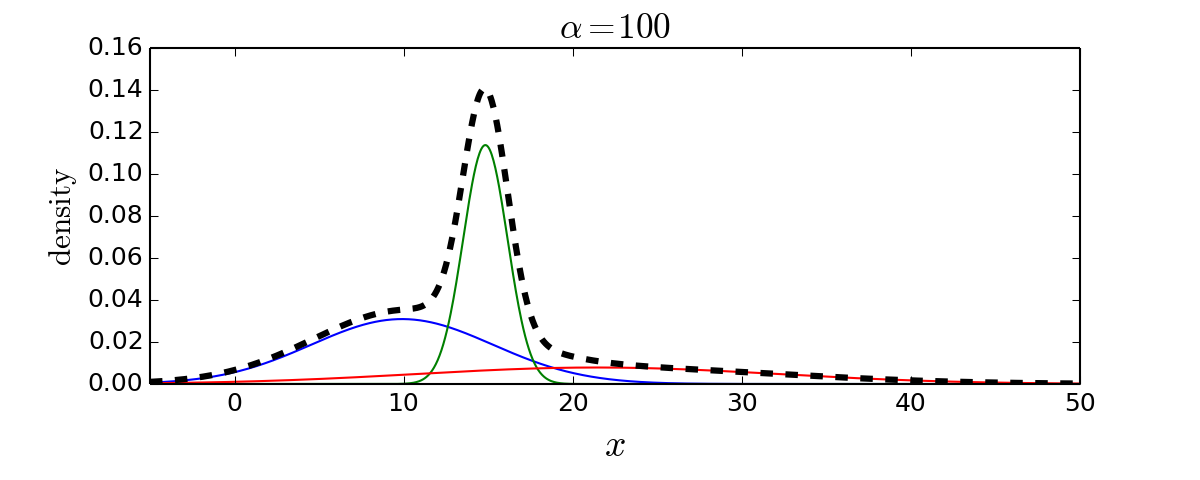}
  \includegraphics[trim=.5cm 0 1.75cm 0, clip, width=0.49\textwidth]{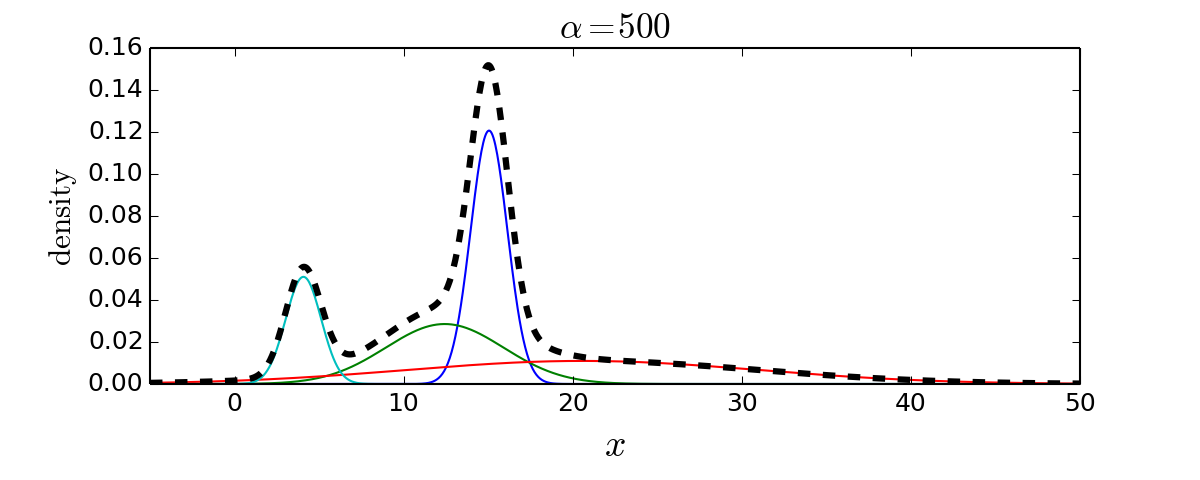}
  \includegraphics[trim=.5cm 0 1.75cm 0, clip, width=0.49\textwidth]{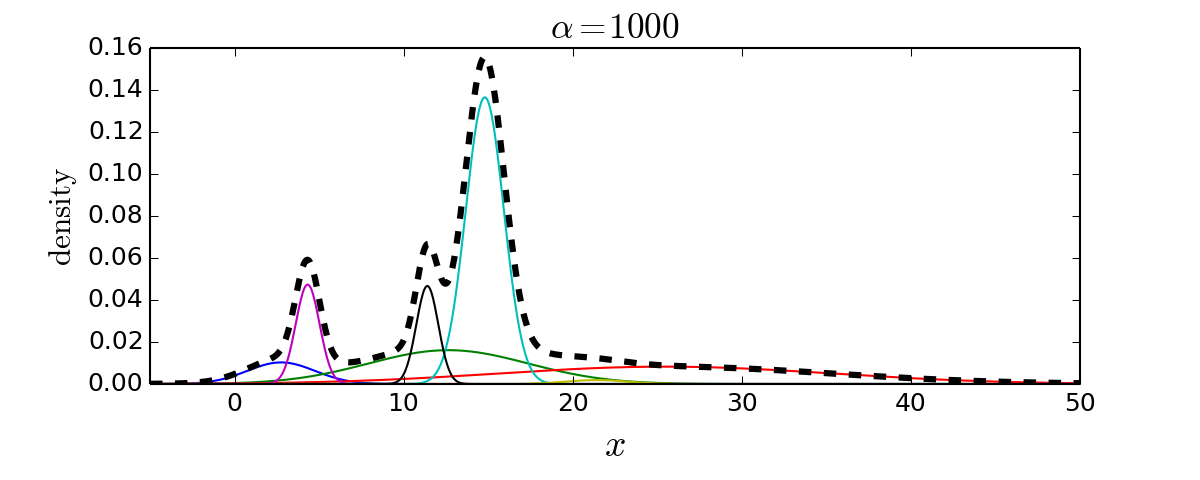}
  \includegraphics[trim=.5cm 0 1.75cm 0, clip, width=0.49\textwidth]{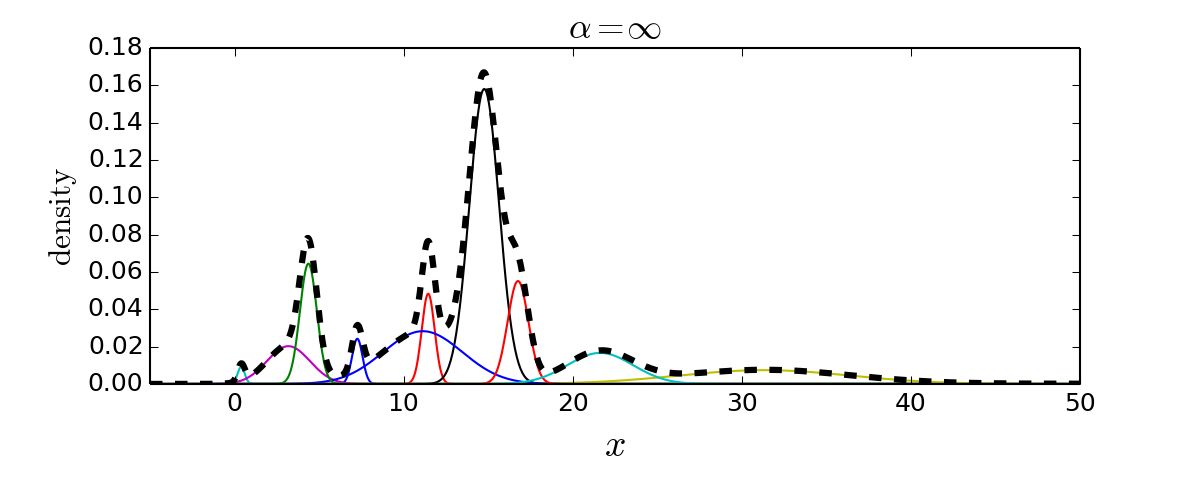}
  \caption{Gaussian mixture with a prior on the number of components $k$, applied to the Shapley galaxy data.
  Top left: Histogram of the data, in units of 1,000 km/s, excluding a small amount of data extending in a tail up to 80,000 km/s.
  Top right: C-posterior on $k$ for a range of $\alpha$ values; $\alpha=\infty$ is the standard posterior.
  Middle and bottom: Mixture density (dotted black line) and components (solid colors) for prototypical samples from the 
  c-posterior, for a range of $\alpha$ values. }
  \label{figure:shapley}
\end{figure}

We use the same model as above, but with $m = 20$ and a data-dependent choice of prior parameters:
$\mu_i\sim\N(\bar x, \hat\sigma^2)$ and $\log(\lambda_i)\sim\N(\log(4/\hat\sigma^2),2^2)$.
For $\alpha\in\{20,100,500,1000\}$, we run the sampler for $10^5$ MH sweeps, with a burn-in of $10^4$ sweeps.
For the standard posterior ($\alpha =\infty$), mixing is considerably slower; 
to compensate, we use $10^6$ sweeps with a burn-in of $2.5\times 10^5$.
This illustrates how inference can be easier under the c-posterior.

As shown in Figure~\ref{figure:shapley}, when $\alpha$ is small, the c-posterior tolerates greater departures from normality,
and uses a smaller number of components to represent the data.
For instance, from a glance at the histogram, one can visually distinguish three or four large groups which appear roughly unimodal,
and when $\alpha$ is around 100 to 500, samples from the c-posterior tend to provide a mixture representation
that corresponds well to these intuitive groups.
For larger values of $\alpha$, additional mixture components are employed,
to account for finer and finer grained aspects of the data distribution.
By the time $\alpha=\infty$, i.e., the standard posterior, the large-scale structures have mostly been fragmented into many small components.

Of course, in a univariate setting like this, one can already visually see the large-scale groups, but clusters in 
high-dimensional data are not so easy to visualize, and having a tool like the c-posterior
to find structures at varying levels of precision may be very useful.
In most applications, the primary use of mixture models is not density estimation, but rather, to provide an interpretable summary of the
data in terms of clusters, and in these cases the c-posterior approach may have much to offer.

\section{Discussion}
\label{section:discussion}


The c-posterior approach seems promising as a general method of robust Bayesian inference,
being both computationally efficient and conceptually well-grounded.
There are a number of directions that would be interesting to pursue in future work.
Further investigation of the small-sample correction is needed---in particular,
assessing its accuracy, and justifying its use in non-discrete cases, if possible.
We have focused on relative entropy due to the computational advantages,
but it would be interesting to explore using other statistical distances,
particularly if fast inference methods could be developed for them as well.
It would be nice if there was a way of inferring the coarsening parameter $\alpha$ from data, somehow;
this would enable one not only to perform robust inference, but also to infer the amount of misspecification,
and ideally, to achieve statistical efficiency when the model is correct.
It would be beneficial, but perhaps difficult,
if precise guarantees could be provided regarding frequentist coverage properties of the c-posterior,
under misspecification. 
Finally, the scope of application of this way of thinking seems potentially larger than Bayesian inference;
it might be interesting to explore adaptations to frequentist procedures.








\appendix

\section{Proofs}
\label{section:proofs}

\subsection{C-posterior for the toy example in Section \ref{section:toy}}
\label{section:toy-details}

Letting $Z = \I(D(\hat p_x||\hat p_X) < R)$, by Bayes' theorem we have that
for $h\in\{\HH_0,\HH_1\}$,
\begin{align*}
\Pi\big(h|Z=1) &\underset{h}{\propto} \Pr(Z=1|h)\Pi(h) \overset{\text{(a)}}{\underset{h}{\propto}} \Pr(Z=1|h) \\
& \overset{\text{(b)}}{=} \E\big(\Pr(Z=1|X_{1:n},h)\,\big\vert\,h)
\overset{\text{(c)}}{=} \E\big(\exp(-\alpha D(\hat p_x||\hat p_X))\,\big\vert\,h)
\end{align*}
where (a) is since $\Pi(h)=1/2$, and (b) is by the law of iterated expectations, and (c) by the fact that $\Pr(R>r) = \exp(-\alpha r)$.
This is easily computed exactly, since, letting $S =\sum_{i = 1}^n X_i = n\hat p_X(1)$, we have
$S|\HH_0\,\sim\,\Binomial(n, 1/2)$ and
$S|\HH_1\,\sim\,\BetaBinomial(n,1,1)=\Uniform\{0,1,\ldots,n\}$.
To derive the approximation in Equation \ref{equation:toy-coarse}, we use Equation \ref{equation:small-sample} below:
$$\E\big(\exp(-\alpha D(\hat p_x||\hat p_X))\,\big\vert\,\theta,h) \approx 
\sqrt{\alpha_n/\alpha}\exp(-\alpha_n D(\hat p_x||p_\theta))
= c\prod_{i = 1}^n p_\theta(x_i)^{\alpha_n/n} $$
where $\alpha_n = 1/(1/n +1/\alpha)$, $p_\theta(x) = \Bernoulli(x|\theta) = \theta^x(1-\theta)^{1-x}$ for $x\in\{0,1\}$,
and $c$ is a constant that does not depend on $\theta$ or $h$.
If $h=\HH_1$, then this yields
\begin{align*}
    \Pr(Z=1|\HH_1) &= \E\Big(\E\big(\exp(-\alpha D(\hat p_x||\hat p_X))\,\big\vert\,\theta,\HH_1)\,\Big\vert\,\HH_1\Big) 
    \approx \E\Big(c\prod_{i = 1}^n p_\theta(x_i)^{\alpha_n/n}\,\Big\vert\,\HH_1\Big) \\
    &= c\int_0^1 \theta^{\alpha_n \bar x}(1-\theta)^{\alpha_n(1-\bar x)}\,d\theta
    = c B\big(1+\alpha_n\bar x,\,1+\alpha_n(1-\bar x)\big).
\end{align*}
If $h=\HH_0$, then $\theta = 1/2$ with probability 1, so $\Pr(Z=1|\HH_0) \approx c(1/2^{\alpha_n/n})^n = c/2^{\alpha_n}$.
Thus, $\Pi(\HH_0|Z=1) = \Pr(Z=1|\HH_0) / (\Pr(Z=1|\HH_0)+\Pr(Z=1|\HH_1)) \approx 1/(1+2^{\alpha_n}B(1+\alpha_n\bar x,\,1+\alpha_n(1-\bar x))$.

\subsection{Justification of the small-sample correction when $|\X|<\infty$.}
\label{secton:small-sample-justification}

Let $\Delta_k =\{p\in\R^k:\sum_i p_i = 1, \, p_i>0 \,\,\forall i\}$.
Let $s\in\Delta_k$. We argue that if $X_1,\dotsc,X_n\iid s$ and $\hat\s_j = \frac{1}{n}\sum_{i = 1}^n\I(X_i = j)$ for $j = 1,\dotsc,k$, then
for $p\in\Delta_k$ near $s$,
\begin{align}
\label{equation:small-sample}
\E \exp(-\alpha D(p\|\hat\s))\approx (n\zeta_n/\alpha)^{\frac{k-1}{2}} \exp(-n\zeta_n D(p\|s)),
\end{align}
where $\zeta_n = (1/n)/(1/n + 1/\alpha)$.  We use bold here to denote random variables.
For $x\in\R^d$, define $C(x)\in\R^{d\times d}$ such that $C(x)_{ij} = x_i\I(i = j)-x_i x_j$, and denote $x'=(x_1,\dotsc,x_{d-1})$.
First, for $q\in\Delta_k$ near $p$,
\begin{align}
\label{equation:relative-entropy-Mahalanobis}
D(p\|q)\approx \tfrac{1}{2}\chi^2(p,q) =\tfrac{1}{2}(p'-q')^\T C(q')^{-1}(p'-q')
\end{align}
by Propositions \ref{proposition:chi-square-relative-entropy} and \ref{proposition:chi-square-Mahalanobis} below.
By the central limit theorem, $\hat\s$ is approximately $\N(s,C(s)/n)$ distributed. Therefore, letting $\q\sim \N(s,C(s)/n)$ and $C = C(s')$,
\begin{align*}
\E\exp(-\alpha D(p\|\hat\s)) &\approx\E\exp(-\alpha D(p\|\q)) \I(\q\in\Delta_k) \\
& \overset{\mathrm{(a)}}{\approx} \E\exp\Big(-\frac{\alpha}{2}(p'-\q')^\T C^{-1} (p'-\q')\Big) \\
& = (2\pi)^{\frac{k-1}{2}}|C/\alpha|^{1/2}\int\N(p'|q',C/\alpha)\N(q'|s',C/n) d q' \\
& \overset{\mathrm{(b)}}{=} (2\pi)^{\frac{k-1}{2}}|C/\alpha|^{1/2}\N(p'|s',(1/\alpha+1/n)C) \\
& = \Big(\frac{1/\alpha}{1/\alpha +1/n}\Big)^{\frac{k-1}{2}} 
     \exp\big(-\tfrac{1}{2}(1/\alpha +1/n)^{-1}(p'-s')^\T C^{-1}(p'-s')\big) \\
& \overset{\mathrm{(c)}}{\approx} (n\zeta_n/\alpha)^{\frac{k-1}{2}} \exp(-n\zeta_n D(p\|s)),
\end{align*}
where (a) is by Equation \ref{equation:relative-entropy-Mahalanobis} along with the approximation $C(\q')\approx C(s')$, (b) uses the convolution formula for independent normals, and (c) is again by Equation \ref{equation:relative-entropy-Mahalanobis}. This yields Equation \ref{equation:small-sample}.

It is well-known that chi-squared distance is a second-order Taylor approximation to relative entropy
\citep[][Lemma 17.3.3]{cover2006elements}; for completeness, we include the proof.

\begin{proposition}
\label{proposition:chi-square-relative-entropy}
For $p,q\in\Delta_k$,
$ D(p\|q) =\tfrac{1}{2}\chi^2(p,q) + o(\|p-q\|^2) $
as $p\to q$, where $D(p\|q) =\sum_i p_i\log(p_i/q_i)$ and $\chi^2(p,q) =\sum_i (p_i-q_i)^2/q_i$.
\end{proposition}
\begin{proof}
Fix $b>0$, and define $f(a) = a\log(a/b)$ for $a>0$. Then by Taylor's theorem, 
\begin{align*}
f(a) &= f(b) + f'(b)(a-b) +\tfrac{1}{2} f''(b)(a-b)^2 + o(|a-b|^2) \\
& = (a-b) +\frac{1}{2}\frac{(a-b)^2}{b} + o(|a-b|^2)
\end{align*}
as $a\to b$. It follows that
\begin{align*}
\sum_{i = 1}^k p_i\log\frac{p_i}{q_i} = \sum_i(p_i-q_i) + \frac{1}{2}\sum_i\frac{(p_i-q_i)^2}{q_i} + o(\|p-q\|^2)
= \tfrac{1}{2}\chi^2(p,q) + o(\|p-q\|^2)
\end{align*}
as $p\to q$.
\end{proof}

The following result expresses the chi-squared distance $\chi^2(p,q)$ in terms of the $(k-1)$-dimensional Mahalanobis distance for $Z'$
when $Z\sim\mathrm{Multinomial}(1,q)$.  For interpretation, note that $C$ below equals $\Cov(Z')$ when $Z\sim\mathrm{Multinomial}(1,q)$.

\begin{proposition}
\label{proposition:chi-square-Mahalanobis}
For any $p,q\in\Delta_k$,
$ \chi^2(p,q) = (p'-q')^\T C^{-1} (p'-q') $
where $C\in\R^{(k-1)\times(k-1)}$ such that $C_{ij} = q_i\I(i = j) - q_i q_j$.
\end{proposition}
\begin{proof}
By the Sherman--Morrison formula for rank-one updates, 
$ C^{-1} =(\diag(q') - q'q'^\T)^{-1} = \diag(q')^{-1} + (1/q_k) {\bm 1}{\bm 1}^\T $
where ${\bm 1} = (1,\dotsc,1)^\T$, hence
$$(p'-q')^\T C^{-1}(p'-q') 
=\sum_{i = 1}^{k-1} \frac{(p_i-q_i)^2}{q_i} + \frac{\big(\textstyle{\sum_{i = 1}^{k-1}} (p_i-q_i)\big)^2}{q_k}$$
and $\sum_{i = 1}^{k-1} (p_i-q_i) = (1-p_k) - (1-q_k) = q_k-p_k$.
\end{proof}

\subsection{Proofs of results in Theory section}
\label{section:theory-proofs}

\begin{proof}[Proof of Lemma \ref{lemma:limit}]
    Since $\Pr(U=V)=0$, we have $\I(U_n<V) \xrightarrow[]{\mathrm{a.s.}} \I(U<V)$,
    and thus, also $W\I(U_n<V) \xrightarrow[]{\mathrm{a.s.}} W\I(U<V)$.
    Hence, by the dominated convergence theorem \citep[][2.44]{Breiman_1968},
    $\Pr(U_n<V) \longrightarrow \Pr(U<V)$
    and
    $$\E\big(W\I(U_n<V)\big)\longrightarrow \E\big(W\I(U<V)\big)$$
    since $0\leq\I(\cdot)\leq 1$, $\,|W\I(U_n<V)|\leq|W|$, and $\E|W|<\infty$.
    By assumption, $\Pr(U<V)>0$, hence $\Pr(U_n<V)>0$ for all $n$ sufficiently large, and
    $$\E(W|U_n<V) = \frac{\E\big(W\I(U_n<V)\big)}{\Pr(U_n<V)} \xrightarrow[n\to\infty]{} 
    \frac{\E\big(W\I(U<V)\big)}{\Pr(U<V)} = \E(W|U<V).$$
\end{proof}

\begin{proof}[Proof of Theorem \ref{theorem:asymptotic}]
    We apply Lemma \ref{lemma:limit} with $U = d(P_\bbtheta,P_o)$, $U_n = d_n(X_{1:n},x_{1:n})$, $V = R$, and $W = h(\btheta)$.
    By assumption, $U_n\xrightarrow[]{\mathrm{a.s.}} U$, $\Pr(U = V) = 0$, $\Pr(U < V) > 0$, and $\E|W| < \infty$. 
    Hence, by Lemma \ref{lemma:limit},
    \begin{align*}
    \E(W\mid U_n<V) &\longrightarrow \E(W\mid U<V) =\frac{\E(W\I(U<V))}{\Pr(U<V)} \\
    & = \frac{\E\big(W\E(\I(U<V)|W,U)\big)}{\E\big(\Pr(U<V\mid U)\big)}
    = \frac{\E(W G(U))}{\E G(U)}
    \end{align*}
    since $V\perp U,W$ by construction. This establishes Equation \ref{equation:asymptotic-posterior-expectation}, and since in particular this holds for any bounded continuous $h$, Equation \ref{equation:asymptotic-posterior} follows.
\end{proof}

\begin{proof}[Proof of Corollary \ref{corollary:weakly-cts-asymptotic}]
Since $X_1,\dotsc,X_n|\btheta\iid P_\bbtheta$ and $x_1,\dotsc,x_n$ behaves like an i.i.d.\ sequence from $P_o$, then
$\hat P_{X_{1:n}}\xLongrightarrow[]{\mathrm{a.s.}} P_\bbtheta$ and $\hat P_{x_{1:n}} \Longrightarrow P_o$
\citep[][Theorem 11.4.1]{dudley2002real}.
Hence, $d_n(X_{1:n},x_{1:n}) \xrightarrow[]{\mathrm{a.s.}} d(P_\bbtheta,P_o)$, and Theorem \ref{theorem:asymptotic} applies.
\end{proof}

\begin{proof}[Proof of Theorem \ref{theorem:continuity}]
    Apply Lemma \ref{lemma:limit} with $U = d(P_\bbtheta,P_o)$, $U_m = d(P_\bbtheta,P_m)$, $V = R$, and $W = h(\btheta)$.
\end{proof}

\bibliographystyle{abbrvnatcaplf}
\bibliography{refs}

\end{document}